\documentclass[onecolumn,journal,a4paper,oneside,final]{IEEEtran}

\usepackage[utf8]{inputenc}
\usepackage[T1]{fontenc}
\usepackage{amsmath}
\usepackage{amssymb}
\usepackage{amsthm}
\usepackage{amsfonts}
\usepackage[dvipsnames]{xcolor}
\usepackage[all]{xy}
\usepackage[hidelinks]{hyperref}
\usepackage{authblk}
\usepackage[nolist]{acronym}
\usepackage{mathrsfs}
\usepackage{bm}
\usepackage{bigstrut}
\usepackage{psfrag}
\usepackage{stackengine}
\usepackage{mathbbol}
\usepackage{lipsum}
\usepackage[sans]{dsfont}
\usepackage{balance}
\usepackage{mathtools,scalerel}
\mathtoolsset{showonlyrefs}

\usepackage{float} 	
\usepackage{array}
\usepackage{wrapfig}
\usepackage{mathtools}
\usepackage{verbatim}   
\usepackage{pgfplots,relsize}
\pgfplotsset{compat=1.17}

\usepackage{tikz}
\usetikzlibrary{shapes,arrows}
\usetikzlibrary{calc}
\usetikzlibrary{patterns} 
\usepackage{mathrsfs}
\usetikzlibrary{fadings}
\usetikzlibrary{shadows.blur}
\usetikzlibrary{spy}
\usetikzlibrary{positioning,chains,fit}

\usepackage{tabularx}
\usepackage{multirow} 	
\usepackage{array, makecell}

\theoremstyle{plain}
\newtheorem{theorem}{Theorem}[section]
\newtheorem{lemma}[theorem]{Lemma}

\newtheorem{corollary}[theorem]{Corollary}

\theoremstyle{definition}
\newtheorem{definition}[theorem]{Definition}

\newtheorem{remark}[theorem]{Remark}
\newtheorem{example}[theorem]{Example}

\usepackage{algorithm}
\usepackage{algpseudocode}

\usepackage{todonotes}

\newcommand{\me}{\mathrm{e}}
\newcommand{\st}{\, | \,} 

\newcommand{\bfc}{\mathbf{c}}
\newcommand{\bfe}{\mathbf{e}}
\newcommand{\bfk}{\mathbf{k}}
\newcommand{\bfp}{\mathbf{p}}
\newcommand{\bft}{\mathbf{t}}
\newcommand{\bfu}{\mathbf{u}}
\newcommand{\bfv}{\mathbf{v}}
\newcommand{\bfx}{\mathbf{x}}
\newcommand{\bfy}{\mathbf{y}}
\newcommand{\bfz}{\mathbf{z}}

\newcommand{\bfE}{\mathbf{E}}
\newcommand{\bfH}{\mathbf{H}}
\newcommand{\bfL}{\mathbf{L}}
\newcommand{\bfX}{\mathbf{X}}
\newcommand{\bfY}{\mathbf{Y}}
\newcommand{\bfzero}{\mathbf{0}}

\newcommand{\veca}{\mathbf{a}}
\newcommand{\vecb}{\mathbf{b}}

\newcommand{\vecu}{\mathbf{u}}
\newcommand{\vecv}{\mathbf{v}}

\newcommand{\card}[1]{\left\vert{#1}\right\vert} 
\newcommand{\floor}[1]{\lfloor{#1}\rfloor} 
\newcommand{\set}[1]{\left\lbrace{#1}\right\rbrace} 

\newcommand{\tendsto}{\longrightarrow}

\newcommand{\intring}{\mathbb{Z}}
\newcommand{\intmodq}[1]{\intring/{#1}\intring}
\newcommand{\units}[1]{(\intmodq{#1})^\times}

\newcommand{\NN}{\mathbb{N}}
\newcommand{\Zqn}{(\intmodq{q})^n}
\newcommand{\orbit}[1]{\mathcal{O}_{#1}}
\newcommand{\divisors}[1]{\mathbb{D}_{#1}}

\newcommand{\prob}{\mathbb{P}} 
\newcommand{\expect}{\mathbb{E}}

\newcommand{\perror}{P_B}
\newcommand{\pep}{\mathsf{PEP}}

\newcommand{\code}{\mathcal{C}}
\DeclareMathOperator{\weight}{wt}
\newcommand{\LW}{\weight_{\scaleto{\mathsf{L}}{4pt}}}
\newcommand{\HW}{\weight_{\scaleto{\mathsf{H}}{4pt}}}
\DeclareMathOperator{\dist}{d}
\newcommand{\LD}{\dist_{\scaleto{\mathsf{L}}{4pt}}}

\newcommand{\KLD}[2]{D(#1 \, || \, #2)} 

\newcommand{\leesphere}[1]{S_{#1}^{(n)}}
\newcommand{\leeball}[1]{V_{#1}^{(n)}}

    \newcommand{\vn}{\mathsf{v}}
    \newcommand{\cn}{\mathsf{c}}

    \newcommand{\neigh}[1]{\ensuremath{\mathcal{N}\left(#1\right)}}
    \newcommand{\msg}[2]{\ensuremath{m_{#1\rightarrow#2}}}
    
    \newcommand{\mchVec}{\ensuremath{\bm{m}_{\vn}}}

    \newcommand{\msgVec}[2]{\ensuremath{\bm{m}_{#1\rightarrow#2}}}
    
    \newcommand{\ens}{\mathscr{C}}

    \newcommand{\permMat}{\bm{\Pi}}

    \newcommand{\thr}{\ensuremath{\delta^\star}}
    \newcommand{\thrSMP}{\ensuremath{\delta^\star_{\scaleto{\mathsf{SMP}}{3.5pt}}}}
    \newcommand{\thrNBP}{\ensuremath{\delta^\star_{\scaleto{\mathsf{BP}}{3.5pt}}}}
    \newcommand{\thrSH}{\ensuremath{\delta^\star_{\scaleto{\mathsf{SH}}{3.5pt}}}}
    
    \DeclareMathOperator*{\argmin}{argmin}
    \DeclareMathOperator*{\argmax}{argmax}
    
    \newcommand{\circConv}{\mathop{\vphantom{\sum}\mathchoice
        {\vcenter{\hbox{\huge $\circledast$}}}
        {\vcenter{\hbox{\Large A}}}{\mathrm{A}}{\mathrm{A}}}\displaylimits}
    \newcommand{\hadProd}{\mathop{\vphantom{\sum}\mathchoice
        {\vcenter{\hbox{\huge $\odot$}}}
        {\vcenter{\hbox{\Large A}}}{\mathrm{A}}{\mathrm{A}}}\displaylimits}

\definecolor{dark_red}{RGB}{150,0,0}
\definecolor{dark_green}{RGB}{0,150,0}
\definecolor{dark_blue}{RGB}{0,0,150}
\definecolor{dark_pink}{RGB}{80,120,90}

\newcommand{\memLC}{memoryless Lee channel}
\newcommand{\constLC}{constant Lee weight channel}
\usepackage{flushend}

\begin{document}
\title{Error-Correction Performance of Regular Ring-Linear LDPC Codes over Lee Channels}

\author{Jessica~Bariffi,~\IEEEmembership{Student Member,~IEEE},
Hannes~Bartz,~\IEEEmembership{Member,~IEEE},
Gianluigi~Liva,~\IEEEmembership{Senior Member,~IEEE} and
Joachim~Rosenthal,~\IEEEmembership{Fellow,~IEEE}
\thanks{
    This paper was presented in part at 2022 IEEE Global Communications Conference.
    
    J.~Bariffi was with the German Aerospace Center, Münchener Strasse 20, D-82234 Wessling, Germany, and is currently with the Technical University of Munich, Theresienstrasse 90, D-80333 München, Germany (e-mail: {\tt jessica.bariffi@tum.de}).
    
    H.~Bartz and G.~Liva are with the German Aerospace Center, Münchener Strasse 20, D-82234 Wessling, Germany (e-mail: {\tt \{hannes.bartz,gianluigi.liva\}@dlr.de}).

    J.~Rosenthal is with the University of Zurich, Institut für Mathematik, Winterthurerstrasse 190, CH-8057 Zürich, Switzerland (e-mail: {\tt rosenthal@math.uzh.ch}).

    J.~Rosenthal has been supported in part by the Swiss National Science Foundation under the grant No. $212865$. J. Bariffi, H. Bartz and G. Liva acknowledge the financial support by the Federal Ministry of Education and Research of Germany in the program of "Souver\"an. Digital. Vernetzt." Joint project 6G-RIC, project identification number: 16KISK022.
}}

\maketitle
 \thispagestyle{empty}
 \IEEEoverridecommandlockouts

 \markboth
    {submitted to IEEE Transactions on Information Theory}
    {J. Bariffi et al.: Performance of Regular Ring-Linear LDPC Codes over Lee Channels}

\begin{abstract}
    Most low-density parity-check (LDPC) code constructions are considered over finite fields. 
    In this work, we focus on regular LDPC codes over integer residue rings and analyze their performance with respect to the Lee metric.
    Their error-correction performance is studied over two channel models, in the Lee metric. 
    The first channel model is a discrete memoryless channel, whereas in the second channel model an error vector is drawn uniformly at random from all vectors of a fixed Lee weight. 
    It is known that the two channel laws coincide in the asymptotic regime, meaning that their marginal distributions match.
    For both channel models, we derive upper bounds on the block error probability in terms of a random coding union bound as well as sphere packing bounds that make use of the marginal distribution of the considered channels. 
    We estimate the decoding error probability of regular LDPC code ensembles over the channels using the marginal distribution and determining the expected Lee weight distribution of a random LDPC code over a finite integer ring. 
    By means of density evolution and finite-length simulations, we estimate the error-correction performance of selected LDPC code ensembles under belief propagation decoding and a low-complexity symbol message passing decoding algorithm and compare the performances. The analysis developed in this paper may serve to design regular \ac{LDPC} codes over integer residue rings for storage and cryptographic application.
\end{abstract}



\begin{acronym}
    \acro{BP}{belief propagation}
    \acro{VN}{variable node}
    \acro{CN}{check node}
    \acro{DE}{density evolution}
    \acro{EXIT}{extrinsic information transfer}
    \acro{i.i.d.}{independent and identically distributed}
    \acro{LDPC}{low-density parity-check}
    \acro{MAP}{maximum a posteriori probability}
    \acro{APP}{a posteriori probability}
    \acro{MCM}{Monte Carlo method}
    \acro{r.v.}{random variable}
    \acro{p.m.f.}{probability mass function} 
    \acro{ML}{maximum likelihood}
    \acro{WEF}{weight enumerating function}
    \acro{MDS}{maximum distance separable}
    \acro{AWE}{average weight enumerator}
    \acro{DMC}{discrete memoryless channel}
    \acro{w.r.t.}{with respect to}
    \acro{BSC}{binary symmetric channel}
    \acro{QSC}{$q$-ary symmetric channel}
    \acro{SMP}{symbol message-passing}
    \acro{RV}{random variable}
    \acro{PMF}{probability mass function}
    \acro{qSC}[$q$-SC]{$q$-ary symmetric channel}
    \acro{RCU}{random coding union}
    \acro{BMP}{binary message-passing}
    \acro{PEG}{progressive edge growth}
    \acro{LSF}{Lee symbol flipping}
    \acro{TV}{total variation}
    \acro{MD}{minimum distance}
\end{acronym}

\section{Introduction}\label{sec:intro}

\IEEEPARstart{T}{he} Lee metric has been introduced in \cite{lee1958some,Ulrich} for phase shift keying modulation purposes, where the first notion of a channel ``matching'' the Lee metric appeared. 
The construction of Lee-metric codes was explored in various contexts (\cite{berlekamp1966negacyclic,chiang1971channels,etzion2010dense,golomb1968algebraic,prange1959use}).
Currently, the Lee metric is considered for applications in post-quantum cryptography (\cite{bariffi2022properties, bariffi2022informationset, chailloux2021classical, cryptoeprint:fuLeeca, weger2020hardness}). It has been shown that the syndrome decoding problem in the Lee metric (originally introduced over $\intmodq{4}$ in \cite{chailloux2021classical}) is NP-hard over any integer residue ring modulo $p^s$, where $p$ is a prime \cite{weger2020hardness}. The paper also provides several generic decoding algorithms to attack the syndrome decoding problem. In \cite{weger2020hardness} the authors showed that the Lee metric information set decoding variants are more costly than their counterparts in the Hamming metric. Therefore, the Lee metric is a promising metric to reduce the key sizes or signature sizes. Furthermore, codes in the Lee metric have potential applications in the context of magnetic \cite{roth1994lee} and DNA \cite{gabrys2017asymmetric} storage systems.

In \cite{bariffi2022properties} a channel model has been introduced over a $q$-ary ring, that adds to the channel input an error vector of given constant Lee weight. We will refer to this channel as the \textit{\constLC}. In the limit of large block length, the single-letter (i.e., marginal) distribution of the error vector elements follows a Boltzmann-like distribution.
This distribution results to be the dominant empirical distribution for vectors of a fixed Lee weight, i.e., under the fixed weight constraints, it is the entropy-maximizing distribution. Introducing an error vector of fixed weight is motivated by code-based cryptosystems, where the error vector is typically generated at the encryption side, with a constant Lee weight. The underlying syndrome decoding problem's hardness is highly dependent on the weight of the error term. As the block length of the code grows large, it is not possible to reduce the Lee weight by a scalar multiplication as shown in \cite{bariffi2022properties}. The decoding and error-correction performance is additionally determined by the minimum distance of a code. Hence, deriving bounds for the minimum distance is an important task. For the Lee metric several analogue bounds to the Hamming metric, such as the Singleton bound, Gilbert-Varshamov bound, the sphere packing bound, have been developed (see \cite{astola1984asymptotic, byrne2023bounds, loeliger1994upper}). These bounds are all with respect to the minimum Lee distance.\\

In this paper, we consider two channel models: The \constLC, and a \ac{DMC} matched to the Lee metric \cite{chiang1971channels}. The first is a channel where a constant-weight error pattern is added to the transmitted codeword, where the error pattern is chosen uniformly at random from the set of vectors with fixed Lee weight and length equal to the block length. It is possible to show that, in the limit of large block length, and with Lee weights that are proportional to the block length, the marginal distribution of the additive error term follows the well-known Boltzmann distribution. The second channel is an additive \ac{DMC}, where the additive error term follows the Boltzmann distribution (however, differently from the constant-weight channel, the Lee weight of the error vector is not fixed). We refer to the second channel as \textit{\memLC}.

Making use of 
the marginal distribution of the channels, we derive upper bounds on the error-correction capability achievable by a code for given block length and rate.
We derive random coding union bounds for both channels as well as a sphere-packing bound over the \memLC, providing a finite-length performance benchmark to evaluate the block error probability of practical coding schemes. In the case of the \memLC, we also derive an upper bound on the decoding failure probability of a general linear block code under \ac{ML} decoding based on the Lee weight distribution (i.e., the Lee distance spectrum) of the code. We compute the average Lee weight distribution of \ac{LDPC} code ensembles \cite{studio3:GallagerBook} over finite integer rings \cite{Fuja05:Ring} and analyze its spectral growth rate. 

Finally, we study the decoding performance of \ac{LDPC} codes over finite integer rings over both channel models. We consider \ac{LDPC} codes over $q$-ary integer residue rings and analyze their performance with respect to the Lee metric from a code ensemble point of view, via density evolution analysis. For simplicity, we focus on regular \ac{LDPC} code ensembles, since this class of \ac{LDPC} code ensembles is mainly used in cryptography. The extension to irregular or protograph-based LDPC code ensembles is straightforward --- the main difference in the analysis being the introduction of variable/check node degree-dependent generating functions in the Lee distance spectrum analysis, as well as degree-dependent density evolution recursions.
The decoding algorithms considered are the well-known \ac{BP} decoding algorithm \cite{DM98, Fuja05:Ring} and the \ac{SMP} algorithm \cite{Lazaro19:SMP}. The \ac{SMP} decoder was originally defined for the $q$-ary symmetric channel. In this work we adapt the decoder to Lee channels accordingly. The performance of both decoders will additionally be compared to the \ac{LSF} decoder presented in \cite{santini2020low}. We provide finite-length simulation results for both the {\memLC} and the {\constLC} for the decoders mentioned. The results are compared to the finite-length performance bounds derived for the corresponding channel model.

The paper is organized as follows. Section \ref{sec:prelim} serves as preliminary section, where we state important definitions and results needed throughout the paper. In Section \ref{sec:boundsChannels}, we derive finite-length bounds on the block error probability achievable by block codes over Lee channel models. In Section \ref{sec:LDPCdistances}, by means of asymptotic enumeration techniques, we derive the average Lee weight spectrum of a regular \ac{LDPC} code ensembles. The Lee weight spectrum serves then to derive bounds on the error probability in Section \ref{sec:LDPCperformance}. We then analyze and compare the performance of \ac{LDPC} codes over both channel models under \ac{BP} and \ac{SMP} decoding. We discuss the main ingredients to adapt the \ac{SMP} from the original setting to the two channel models in the Lee metric, which relies on an assumption for the extrinsic channel probability. We justify this assumption using empirical results. Finally, conclusions follow in Section \ref{sec:conc}. 

\section{Preliminaries}\label{sec:prelim}
In this section we introduce the basic notation and results required in the course of the paper. In the following we denote by $\intmodq{q}$ the ring of integers modulo $q$, where $q$ is a positive integer. For simplicity, we assume that $\intmodq{q}$ is represented by the set $\set{0, 1, \dots, q-1}$. The set of units of $\intmodq{q}$ will be denoted by $\units{q}$. By abuse of notation we will call an element of $\Zqn$ a vector of length $n$ and we will denote it by bold lower case letters. Similarly, matrices are denoted by bold upper case letters. For any real number $x$, we use the notation $[\, x \, ]^+ := \max(0, x)$.  We denote by $X$ a random variable over a discrete alphabet $\mathcal{X}$ and let $x \in \mathcal{X}$ be its realization. For every $x \in \mathcal{X}$, we will denote the probability distribution of $X$ by $$P_X(x) := \prob(X = x).$$

Given a positive integer $n$ and an $s$-tuple of nonnegative integers $\bfk :=(k_1, \ldots, k_s)$ satisfying $\sum_{i = 1}^s k_i = n$, we denote the multinomial coefficient by
\begin{align*}
    \binom{n}{\bfk} := \binom{n}{k_1, \ldots, k_s} = \frac{n!}{k_1! \ldots k_s!}.
\end{align*}

\subsection{The Lee Metric}\label{subsec:leeMetric}

\begin{definition}\label{def:leeweight}
    Let $a \in \intmodq{q}$. Its \textit{Lee weight} is defined as
	\begin{align*}
		\LW(a) := \min (a, q-a).
	\end{align*}
	For a vector $\veca = (a_1, \dots, a_n) \in \Zqn$ of length $n$, its Lee weight is defined to be the sum of the Lee weights of its entries, i.e.
	\begin{align*}
		\LW(\veca) = \sum_{i=1}^n \LW(a_i).
	\end{align*}
\end{definition}

Intuitively, we can view the elements of $\intmodq{q}$ on a circle with equal distances between them. Then the Lee weight of $a \in \intmodq{q}$ is the minimal number of arcs separating $a$ from the origin $0$. This yields the following symmetry property of the Lee weight,
\begin{align}\label{eq:symmetry_Lee}
    \LW(a) = \LW(q-a).
\end{align} 
Equation \eqref{eq:symmetry_Lee} implies that the Lee weight of any element in $\intmodq{q}$ can never exceed $\floor{q/2}$. Furthermore, we observe that the Lee weight of $a \in \intmodq{q}$ is always lower bounded by its Hamming weight, denoted by $\HW (a)$, which is equal to $1$ if $a$ is nonzero and equal to zero otherwise. Hence, similarly we have $\LW(a) \leq \HW(a)\floor{q/2}$. Equality between the two weights holds if and only $q \in \set{2, 3}$ for every choice of $a$. Hence, for a vector $\veca \in \Zqn$ we have
\begin{align*}
    \HW(\veca) \leq \LW(\veca) \leq n\cdot \floor{q/2}.
\end{align*}

Similar to the Hamming weight, the Lee weight induces a distance between two vectors.
\begin{definition}\label{def:leeDistance}
    Let $\veca$ and $\vecb$ be two vectors in $\Zqn$. The \textit{Lee distance} between $\veca$ and $\vecb$ is the Lee weight of their difference, i.e.
    \begin{align*}
        \LD (\veca, \vecb) := \LW(\veca-\vecb).
    \end{align*}
\end{definition}
It is easy to show that the Lee distance is a metric over the finite ring of integers $\intmodq{q}$.\\

Lemma \ref{lem:expectLW} shows the expected Lee weight of a randomly chosen element in $\intmodq{q}$.
\begin{lemma}[\cite{wyner1968upper}]\label{lem:expectLW}
    Let $A$ be a uniformly distributed random variable over $\intmodq{q}$. The expected Lee weight of $A$ is
    \begin{align}\label{equ:expectLW}
        \delta_q := \expect\left(\LW(A)\right) =
        \begin{cases}
            \left(q^2-1\right)/{4q} & \text{if } q \text{ is odd},\\
            q/4 & \text{if } q \text{ is even}.
        \end{cases}
    \end{align}
\end{lemma}
\begin{proof}
    As $A$ is chosen uniformly at random from $\intmodq{q}$, we have $\prob(A = i) = \frac{1}{q}$ for every $i \in \intmodq{q}$. Then,
    $$\expect(\LW(A)) = \sum_{i = 0}^{q-1} \LW(i) \prob(A = i) = \frac{1}{q} \sum_{i = 0}^{q-1} \LW(i),$$
    where the summation over the Lee weights can be represented as a sum of integers. In fact, if $q$ is odd, then $\sum_{i = 0}^{q-1} \LW(i) = 2 \sum_{i = 1}^{(q-1)/2} i$, whereas for $q$ even we have $\sum_{i = 0}^{q-1} \LW(i) =q/2 +  2 \sum_{i = 1}^{q/2 - 1} i$. Applying the formula for the sum of the first $n$ integers, for $n \in \mathbb{N}$, yields the desired statement.
\end{proof}

Let us define now the $n$-dimensional Lee sphere, $\leesphere{t, q}$, (respectively the $n$-dimensional Lee ball, $ \leeball{t, q}$) over $\intmodq{q}$ centered at the origin of radius $t$ by
\begin{align*}
    \leesphere{t, q} &:= \set{\bfx \in \Zqn \st \LW(\bfx) = t }\\
    \leeball{t, q} &:= \set{\bfx \in \Zqn \st \LW(\bfx) \leq t }.
\end{align*}

\begin{lemma}\cite[Lemma 1]{bariffi2022properties}\label{lem:marginal}
    Assume that $\bfx \in (\intmodq{q})^n$ has been drawn uniformly at random among all vectors of Lee weight $t$.
    Let $X$ denote the random variable defining the realizations of an entry $x$ of $\bfx$. As $n$ grows large, for every $i \in \intmodq{q}$, the probability of $X$ taking the value $i$ is given by
    \begin{align}\label{equ:marginal}
        \prob(X = i) = \frac{1}{Z(\beta)} \exp{(-\beta \LW(i))}
    \end{align}
    where $\beta$ is the unique real solution to the weight constraint $t/n = \sum_{i = 0}^{q-1} \LW(i) \prob (X = i)$ and $Z(\beta)$ denotes the normalization constant.
\end{lemma}
In the following, let $\delta := t/n$ be the normalized Lee weight. Note that if $\delta = \delta_q$, then $X$ is distributed uniformly over $\intmodq{q}$ and hence $\beta = 0$. Moreover, $\beta > 0$ if and only if $\delta < \delta_q$. The distribution in \eqref{equ:marginal} is closely related to the Boltzmann distribution \cite{CoverThomasBook, boltzmann1868studien}. The Boltzmann distribution gives the probability that a system will be in a certain state depending on that states' energy and temperature. In statistical mechanics the distribution is used for systems of fixed compositions all being in a thermal equilibrium. Additionally, the distribution maximizes the entropy subject to a mean energy state. In our case the Lee weight may be interpreted as the energy value of a state $e \in \intmodq{q}$. Hence, we will refer to the distribution in \eqref{equ:marginal} as \textit{Boltzmann distribution} and we will denote it by $B_\delta$.

Finally, we introduce the normalized logarithmic surface (respectively, volume) spectra
\begin{align*}
    \sigma_{\delta n}^{(n)} &:= \frac{1}{n}\log_2 \left(\card{\leesphere{\delta n, q}}\right) \quad \text{and} \\
    \nu_{\delta n}^{(n)} &:= \frac{1}{n}\log_2 \left(\card{\leeball{\delta n, q}}\right)
\end{align*}
while their asymptotic counterparts are denoted by
\begin{align*}
    \sigma_{\delta} &:= \lim_{n \to \infty} \frac{1}{n}\log_2 \left(\card{\leesphere{\delta n, q}}\right) \quad \text{and} \\
    \nu_{\delta} &:= \lim_{n \to \infty} \frac{1}{n}\log_2 \left(\card{\leeball{\delta n, q}}\right).
\end{align*}

\subsection{Information-Theoretic Definitions}

The entropy of $X$ is then defined to be
\begin{align}\label{equ:entropy}
    H(X) := H(P_X) = -\sum_{x \in \mathcal{X}}P_X(x) \log_2 P_X(x),
\end{align}
where by convention for $P_X(x) = 0$ we set $P_X(x)\log_2P_X(x) = 0$. We will make use of the following, well-known result.
\begin{theorem}\cite[Theorem 2.5.1]{CoverThomasBook}\label{thm:chainRule}
    Let $X_1, \dots, X_n$ be a sequence of random variables drawn according to a probability distribution $P(x_1, \dots , x_n)$. Then the entropy of the sequence satisfies
    \begin{align}
        H(X_1, \dots , X_n) 
        &= \sum_{i=1}^n H(X_i \st X_{i-1}, \dots, X_1)\\
        &\leq \sum_{i=1}^n H(X_i)
    \end{align}
    where equality holds if and only if the $X_i$ are independent.
\end{theorem}

Given an empirical distribution $\bfp :=(p_1, p_2, \ldots, p_{\card{\mathcal{X}}})$ we have that  that for any positive integer $n$ (see \cite[Theorem 11.1.3]{CoverThomasBook})
\begin{align}\label{equ:bounds_multinomial} 
    \frac{1}{(n+1)^{\card{\mathcal{X}}}} 2^{nH(\bfp)} \leq \binom{n}{n\bfp} \leq  2^{nH(\bfp)}.
\end{align}
Consider two probability distributions $P_X(x)$ and $\widetilde{P}_X(x)$ over a shared alphabet $\mathcal{X}$. Their Kullback-Leibler divergence is denoted as
    \begin{align}\label{equ:def:relative_entropy}
    	\KLD{P_X}{\widetilde{P}_X}
    	    := 
    	\sum_{x\in \mathcal{X}}P_X(x) \log_2\left( \frac{P_X(x)}{\widetilde{P}_X(x)} \right).
    \end{align} 

An alternative measure of the similarity of two probability distributions is the total variation distance. We define the distance only for discrete probability distributions, since this paper only deals with this case. We follow the description of \cite[Proposition 5.2]{wilmer2009markov} and define the total variation distance between two distributions $P_X$ and $\widetilde{P}_X$ over $\mathcal{X}$ as
\begin{align*}
    \mathsf{TV}(P_X, \widetilde{P}_X) := \frac{1}{2} \sum_{x \in \mathcal{X}} \card{ P_X(x) - \widetilde{P}_X(x) }.
\end{align*}

\subsection{Low-Density Parity-Check Codes over Finite Integer Rings}
We will now introduce linear codes over integer residue rings $\intmodq{q}$ and we will focus in particular on \ac{LDPC} codes over $\intmodq{q}$.
A \textit{ring-linear code} $\code \subset (\intmodq{q})^n$ is a $\intmodq{q}$-submodule of $(\intmodq{q})^n$. Similar to codes over finite fields, codes over rings have a length given by $n$ and a $\intmodq{q}$-dimension given by $k := \log_q(\card{\code})$. We then refer to $\code$ as $[n, k]$ linear code over $\intmodq{q}$. A code $\code$ can be represented by the kernel of a parity-check matrix $\bfH \in (\intmodq{q})^{m\times n}$ with $m \geq n-k$, i.e.,
\begin{align*}
    \code = \set{ \bfx \in (\intmodq{q})^n \st \bfH \bfx^{\top} = \bfzero^{\top} }. 
\end{align*}

We denote by $M := \card{\code}$ the number of codewords in $\code$. Recall that the code rate of an $[n, k]$ linear code $\code$ of size $M$ over $\intmodq{q}$ is given by
\begin{align}
    R_2 &= \frac{\log_2 M}{n} \qquad \text{bits per channel use}\\
\intertext{or}
    R &= \frac{\log_q M}{n} \qquad \text{symbols per channel use}.
\end{align}
depending on the choice of logarithm's base.

Different to codes over finite fields, a ring-linear code $\code \subseteq (\intmodq{q})^n$ does not always admit a basis. If it admits a basis, we call $\code$ a \textit{free code}. In that case, a parity-check matrix $\mathbf{H}$ is of size $(n-k) \times n$.
\ac{LDPC} codes \cite{studio3:GallagerBook} are binary linear error-correcting codes characterized by a sparse parity-check matrix. In \cite{Fuja05:Ring}, the authors analyzed \ac{LDPC} codes over $\intmodq{q}$ defining the nonzero entries of the parity-check matrix over the units $\units{q}$. Restricting to only unit elements as nonzero entries of a parity-check matrix immediately implies that the code is free. In this way, when randomly drawing the nonzero entries of a parity-check matrix among the unit elements, we assure that the resulting LDPC code is always free. This is not always true when drawing the nonzero entries also among the zero divisors. Additionally, allowing zero divisors in a parity-check matrix can yield to the situation where there is a column whose nonzero entries are all zero divisors. Over $\intmodq{p^s}$, for a prime number $p$, the resulting code would have a minimum Lee distance of at most $\floor{p^s/2}$ (see \cite{Fuja05:Ring} for more details). In the following, let $\code$ always denote an $[n, k]$ linear block code over $\intmodq{q}$ and let $\bfH \in (\intmodq{q})^{m \times n}$ be a parity-check matrix of $\code$, where $m \geq n-k$ and where $m = n-k$ if and only if the code $\code$ is free. A parity-check matrix $\bfH$ can be described by a bipartite graph $\mathcal{G} = (\mathcal{V}, \mathcal{E})$ consisting of a set of vertices $\mathcal{V}$ and a set of edges $\mathcal{E}$ connecting the vertices. The set of vertices consists of two disjoint sets: the set of \acp{VN} $\set{\vn_1, \dots, \vn_{n}}$, representing the columns of $\bfH$, and the set of \acp{CN} $\set{\cn_1, \dots , \cn_{m}}$, representing the rows of $\bfH$. A variable node $\vn_j$ is connected to a check node $\cn_i$ by an edge if and only if the corresponding entry $h_{ij}$ in the parity-check matrix is nonzero. The edge carries as label the entry $h_{ij}$.
The degree $d_\vn$ of a variable node $\vn$ is the number of edges connected to $\vn$. The neighbors $\mathcal{N}(\vn)$ of a variable node $\vn$ is the set of check nodes connected to $\vn$. Similarly, we define the degree $d_\cn$ and the neighbors $\mathcal{N}(\cn)$ of a check node $\cn$.

We consider regular nonbinary \ac{LDPC} codes which have a constant variable node degree $d_\vn = d_v$ and a constant check node degree $d_\cn = d_c$.  We denote by $\ens_{d_v,d_c}^n$ the unstructured regular \ac{LDPC} code ensemble of length $n$, i.e. the set of all \ac{LDPC} codes defined by an $(m \times n)$ parity-check matrix, whose associated bipartite graph has constant variable node degree $d_v$ and constant check node degree $d_c$. This ensemble has then the designed rate $R_0 = 1 - m/n$. As proposed in \cite{Fuja05:Ring}, when sampling an \ac{LDPC} code from $\ens_{d_v, d_c}^n$, we assume that the nonzero entries of a parity-check matrix are drawn independently and uniformly at random from the set of units $\units{q}$.

\subsection{Message Passing Decoders}\label{subsec:messagepassing}
We briefly recall two message passing algorithms for nonbinary \ac{LDPC} codes. The first algorithm is the well-known (nonbinary) \ac{BP} algorithm. The second algorithm is a message passing algorithm where the messages exchanged between variable and check nodes are hard symbol estimates. The latter algorithm, dubbed \ac{SMP}, generalizes the Gallager-B algorithm \cite{studio3:GallagerBook} and \ac{BMP} algorithm \cite{lechner2011analysis} to nonbinary alphabets.

Let us fix some notation used in the description of the two decoders. We denote by $\msgVec{\vn}{\cn}$ the message sent from variable node $\vn$ to a neighboring check node $\cn$ and vice versa $\msgVec{\cn}{\vn}$ is the message sent from $\cn$ to $\vn$. Furthermore, we will denote the likelihood at the variable node $\vn$ input (associated with the corresponding channel observation) by
\begin{align*}
    \mchVec:=\left(P_{Y|X}(y\mid 0),\dots,P_{Y|X}(y\mid q-1)\right),
\end{align*}
i.e., this is the vector of probabilities of the channel output $y$, conditioned on the $q$ possible channel input values. For every connected variable node $\vn$ and check node $\cn$ we denote by $h_{\cn\vn}$ the corresponding entry in the parity-check matrix $\bfH$. Note, since the nonzero entries of $\bfH$ were chosen to be units modulo $q$, the inverse $h_{\cn\vn}^{-1}$ is guaranteed to exist.

\subsubsection{Belief Propagation Decoding}\hfill\\
We consider now the \ac{BP} algorithm for nonbinary \ac{LDPC} codes over finite rings. The decoder consists of four main steps that are outlined below, where Step \ref{itm:CN-to-VN-NBP} and \ref{itm:VN-to-CN-NBP} are repeated at most $\ell_{\max}$ times. For every connected variable node $\vn$ and check node $\cn$ we let $\permMat_{\cn \vn}$ be the $(q \times q)$ permutation matrix induced by $h_{\cn\vn}$.\footnote{In fact, consider two random elements $X,X' \in \intmodq{q}$ with probability distributions $P(X)$ and $Q(X')$, respectively. Denote by $\bm{m} = (P(0), P(1), \ldots, P(q-1))$ and by $\bm{m}' = (Q(0), Q(1), \ldots, Q(q-1))$. If $X'=h X$ for $h \in \units{q}$, then since $Q(X') = P(hX)$ the distribution $\bm{m}'$ is obtained by permuting $\bm{m}$ with a permutation that is completely determined by the multiplication coefficient $h$.}
\begin{enumerate}
    \item \textbf{Initialization.}
        Each variable node $\vn$ receives the channel observation in the form of $\mchVec$. Then, the variable node $\vn$ sends to each $\cn \in \mathcal{N}(\vn)$ the permuted channel observation, i.e.,
        \begin{align*}
            \msgVec{\vn}{\cn} = \mchVec \cdot \permMat_{\cn \vn} .
        \end{align*}
        
    \item\label{itm:CN-to-VN-NBP} \textbf{\ac{CN}-to-\ac{VN} step.} 
        Consider a given check node $\cn$ and a neighboring variable node $\vn \in \mathcal{N}(\cn)$. For the message $\msgVec{\cn}{\vn}$, the check node computes the circular convolution $\circConv$ of the incoming messages $\msgVec{\vn'}{\cn}$ from all neighboring variable nodes $\vn'\in \mathcal{N}(\cn) \setminus \set{\vn}$ as
        \begin{align*}
            \vecu=\circConv_{\vn'\in\neigh{\cn}\setminus\{\vn\}}\msgVec{\vn'}{\cn}
        \end{align*}
        and sends to every neighboring variable node $\vn \in \mathcal{N}(\cn)$ a permuted version of $\vecu$ according to the permutation $\permMat_{\cn \vn}^{-1}$, i.e., the \ac{CN}-to-\ac{VN} message is
        \begin{equation}
            \msgVec{\cn}{\vn}=\vecu\cdot\permMat_{\cn \vn}^{-1}.
        \end{equation}

    \item\label{itm:VN-to-CN-NBP} \textbf{\ac{VN}-to-\ac{CN} step.} 
        The variable node $\vn$ computes the Schur product $\odot$ of all incoming messages but the one from check node $\cn$ and normalizes the result by a constant $K$ (to obtain a proper probability vector)
        \begin{equation}
            \vecv = K \hadProd_{\cn'\in\neigh{\vn}\setminus\{\cn\}}\msgVec{\cn'}{\vn}.
        \end{equation}
        Finally, it applies the permutation matrix $\permMat_{\cn \vn}$ to the vector $\vecv$ and sends the following message to the check node $\cn$
        \begin{equation}
            \msgVec{\vn}{\cn}=\vecv\cdot\permMat_{\cn \vn}. 
        \end{equation}

    \item \textbf{Final decision.}
        The final decision happens at the variable node side. After at most $\ell_{\max}$ iterations of steps \ref{itm:CN-to-VN-NBP} and \ref{itm:VN-to-CN-NBP} each variable node computes the Schur product of all incoming messages, yielding the \ac{APP} estimate
        \begin{align*}
            \mchVec^{\scriptscriptstyle \textsf{APP}}=\hadProd_{\cn\in\neigh{\vn}}\msgVec{\cn}{\vn}.
        \end{align*}
        The decision $\hat{x}$ is the index of the maximal entry of $\mchVec^{\scriptscriptstyle \textsf{APP}}$
        \begin{equation}
            \hat{x}=\underset{i\in\intmodq{q}}{\arg\max}\, m_{\vn,i}^{\scriptscriptstyle \textsf{APP}}.
        \end{equation}
\end{enumerate}

\subsubsection{Symbol Message Passing Decoding}\hfill\\
 The \ac{SMP} algorithm is a message-passing algorithm for nonbinary \ac{LDPC} codes, where each message exchanged by a variable node/check node pair is a symbol, i.e., a hard estimate of the codeword symbol associated with the variable node. Following the principle outlined in \cite{lechner2011analysis}, the messages sent by check nodes to variable nodes are modeled as observations at the output a $q$-ary input, $q$-ary output \ac{DMC}. By doing so, the messages at the input of each variable node can be combined by multiplying the respective likelihoods (or by summing the respective log-likelihoods), providing a simple update rule at the variable nodes.

Assume we have a \ac{DMC} over $\intmodq{q}$ with output $w$ and channel law $P_{W \st X}(w \st x)$. We define the log-likelihood of $w$ given $x$ by $L_x(w) := \log \left(P_{W \st X}(w \st x)\right)$ and the log-likelihood vector by
\begin{align*}
    \bfL(w):=(L_0(w),L_1(w),\dots,L_{q-1}(w)).
\end{align*}
With a slight abuse of notation, we will use the $\bfL(\cdot)$ for different channels, where the channel law to be applied is made clear by the argument.

\begin{enumerate}
    \item \textbf{Initialization.}
        The decoder is initialized by forwarding the channel observation $y$ to every variable node $\vn$. Then, the variable node $\vn$ sends to each $\cn \in \mathcal{N}(\vn)$
        \begin{align*}
            \msg{\vn}{\cn} = y  .
        \end{align*}
        
    \item \textbf{\ac{CN}-to-\ac{VN} step.}
     Consider a given check node $\cn$ and a neighboring variable node $\vn \in \mathcal{N}(\cn)$. For the message $\msgVec{\cn}{\vn}$, the check node computes
        \begin{equation}
            \msg{\cn}{\vn} =h_{\cn,\vn}^{-1}\sum_{\vn'\in\neigh{\cn}\setminus\{\vn\}} h_{\cn,\vn'}\msg{\vn'}{\cn}.
        \end{equation}
        
    \item \textbf{\ac{VN}-to-\ac{CN} step.}
        At each variable node $\vn$, incoming messages are treated as observations of the codeword symbol at the output of an ``extrinsic channel'' (\cite{lechner2011analysis, Ashikhmin:AreaTheorem}) modelled as a \acf{qSC} with error probability $\xi \in [0, 1]$, i.e., with conditional probability 
        \begin{align}
             P_{M \st X}(m \st x)=
            \begin{cases}
                1-\xi & \text{if } m = x
                \\ 
                \xi/(q-1) & \text{otherwise}
            \end{cases}.\label{eq:QSCapproxdec}
        \end{align}
        For the calculation of the message to be sent of each check node $\cn \in \mathcal{N}(\vn)$, \eqref{eq:QSCapproxdec} is used to compute the log-likelihood vector
        \begin{equation}\label{eq:def_E_vec}
            \bfE=\bfL(y)+\sum_{\cn'\in\neigh{\vn}\setminus\{\cn\}}\bfL\left(\msg{\cn'}{\vn}\right).
        \end{equation}
        For each $\cn \in \mathcal{N}(\vn)$, the message sent by the variable node $\vn$ is then
        \begin{equation}
            \msg{\vn}{\cn}=\underset{i\in\intmodq{q}}{\arg\max} \, E_i.
        \end{equation}
    
    \item \textbf{Final decision.}
        After at most $\ell_{\max}$ iterations for each variable node $\vn$ we compute 
        \begin{equation}\label{eq:final}
            \bfL^{\scriptscriptstyle \textsf{FIN}}=\bfL(y)+\sum_{\cn\in\neigh{\vn}}\bfL\left(\msg{\cn}{\vn}\right).
        \end{equation}
        Then the final decision, $\hat{x}$, is the index of the maximal entry of $\bfL^{\scriptscriptstyle \textsf{FIN}}$, i.e.,
        \begin{equation}
            \hat{x}=\underset{i\in\intmodq{q}}{\arg\max}\, L_i^{\scriptscriptstyle \textsf{FIN}}.
        \end{equation}
\end{enumerate}

Note that the extrinsic channel of \eqref{eq:QSCapproxdec} is modelled as a \ac{qSC} with error probability $\xi$. As it will be shown in Section \ref{subsec:DE}, this choice yields an accurate description of the extrinsic channel conditional probability, despite of its simplicity. The extrinsic channel parameter $\xi$ is iteration-dependent. Its evaluation can be performed via Monte Carlo simulations, or by using estimates that follow from \ac{DE} analysis \cite{lechner2011analysis,Lazaro19:SMP}.


\subsection{Lee Channels}\label{sec:channel}

We consider classical additive channel models over the $q$-ary alphabet $\intmodq{q}$.
In the \constLC, a random error vector $\bfe \in (\intmodq{q})^n$ of fixed Lee weight $\LW(\bfe) = t$ is added to the channel input $\bfx \in \code$, i.e., the channel output is
\begin{align*}
    \bfy = \bfx + \bfe.
\end{align*}
More specifically, the error vector $\bfe$ is drawn uniformly at random over the Lee sphere $\leesphere{t, q}$ of radius $t = \delta n$. 
Hence, the channel transition probability for the {\constLC} is 
\begin{align}\label{eq:transProb_constant}
    P_{\mathbf{Y} \st \mathbf{X}}(\bfy \st \bfx)
        =
    \begin{cases}
		\card{\leesphere{\delta n, q}}^{-1} & \text{if } \LD(\bfy,\bfx) = \delta n,\\
		0					& \text{otherwise}. 
	\end{cases}
\end{align}
As a consequence of Lemma \ref{lem:marginal}, the marginal distribution of the error terms follows, for large $n$, the Boltzmann distribution
\begin{equation}
    P_{E}(e) = \frac{1}{Z(\beta)}\exp\left(-\beta\LW(e)\right) \label{eq:Lee_channel}
\end{equation}
where $\beta$ follows by enforcing $\expect\left(\LW(E)\right)=\delta$. In the course of this paper we will denote the Boltzmann distribution from Equation \eqref{eq:Lee_channel} by $B_{\delta}$.

The {\memLC} is a \ac{DMC} defined by 
\begin{align}\label{equ:DMC}
    y = x + e 
\end{align}
where $y,x,e \in \intmodq{q}$, and where the probability distribution of $e$  matches the marginal distribution of the {\constLC} of \eqref{eq:Lee_channel}.
We denote the expected normalized Lee weight of $\bfe \in \intmodq{q}$ by 
$$\delta = \expect\left(\frac{1}{n} \LW(\bfE)\right).$$
As the channel is memoryless, we have again $\delta = \expect\left(\LW(E)\right)$.


\section{Finite-Length Bounds for Lee Channels}\label{sec:boundsChannels}
In this section we are going to derive bounds on the error probability achievable by an $[n, k]$ code over both the {\constLC} and the {\memLC} defined in Section \ref{sec:channel}.\footnote{Several of bounds minimum distance achievable by $[n,k]$ codes can be found in \cite{astola1984asymptotic, byrne2023bounds, loeliger1994upper}}. In the first case we will see an achievability bound in terms of a random coding union bound. For the {\memLC} we will derive an upper bound again in terms of a \ac{RCU} bound as well as a converse bound, meaning a lower bound, achievable by any $[n, k]$ code in terms of a sphere-packing bound. 

For both channel models we distinguish between \ac{ML} decoding and \ac{MD} decoding, that is, given a received word $y \in \intmodq{q}$, we consider the \ac{ML} decoding rule
\begin{align}\label{eq:ML}
    \hat{\bfx}_{\mathrm{ML}} = \argmax_{\bfx\in \code} P_{\mathbf{Y} \st \mathbf{X}}(\bfy \st \bfx) 
\end{align}
and the \ac{MD} decoding rule
\begin{align}\label{eq:MD}
	\hat{\bfx}_{\mathrm{MD}} = \argmin_{\bfx \in \code}\LD(\bfy, \bfx).
\end{align}
Note that the two decoding rules coincide over the {\memLC} for $\delta \leq \delta_q$. In the \constLC, the \ac{ML} decoder gives a list of all codewords which are at distance $\delta n$ from the received word $\bfy$ and it outputs one of the codewords in this list randomly. Hence, the two decoding rules for the {\constLC} coincide whenever $\delta n$ is within the decoding radius of the code $\code$.

\subsection{Bounds on the Lee Spheres and Lee Balls}\label{subsec:growthrate}
Before proceeding with the derivation of the error probability bounds, we  first derive upper bounds on the size of a Lee sphere and a Lee ball, respectively. 

We denote by $H_\delta := H(B_\delta)$ the entropy of the Boltzmann distribution with parameter $\delta$ and we introduce the notation
\begin{align*}
    H_\delta^+ := 
    \begin{cases}
        H_\delta & 0 \leq \delta \leq \delta_q \\
        \log_2(q) & \delta_q < \delta < r.
    \end{cases}
\end{align*}

\begin{lemma}[Growth rate of the surface spectrum]\label{lemma:bound_mu_H}
    For any positive integer $\delta n$ the surface spectrum is upper bounded by
    \begin{align}
        \sigma_{\delta n}^{(n)} \leq H_\delta .
    \end{align}
    In particular, as $n$ grows large it holds that $\sigma_\delta =  H_\delta$.
\end{lemma}
\begin{proof}
    Let $\bfX = ( X_1, \dots, X_n)$ be a finite sequence of random variables $X_i$ chosen uniformly at random in the Lee sphere $\leesphere{\delta n, q}$. Since $\bfX$ is uniformly distributed in the sphere, its entropy is given by $H(\bfX) = \log_2\left(\card{\leesphere{\delta n, q}}\right)$. Hence, the normalized logarithmic surface area is
    \begin{align*}
        \sigma_{\delta n}^{(n)} = \frac{1}{n} H(\mathbf{X}).
    \end{align*}
    The chain rule for the entropy, Theorem \ref{thm:chainRule}, and the fact that the $X_i$´s are identically distributed, yield
    \begin{align*}
        H(\mathbf{X}) \leq \sum_{i=1}^n H(X_i) = n H(X_1).
    \end{align*}
    Since the Boltzmann distribution $B_\delta$ is the distribution of $X_1$ maximizing the entropy under the constraint that $\expect(\LW(X_1)) = \delta$, the desired upper bound follows.
    To get the asymptotic result it suffices to take limits on both sides of the inequality.
\end{proof}

\begin{lemma}[Growth rate of the volume spectrum]\label{lemma:bound_vol_H}
    For any positive integer $\delta n$ the volume spectrum is upper bounded by
    \begin{align}
        \nu_{\delta n}^{(n)} \leq H_\delta^+ .
    \end{align}
    In particular, as $n$ grows large we have that $\nu_\delta =  H_\delta^+$.
\end{lemma}
\begin{proof}
    The proof follows in a similar fashion to the proof of the growth rate of the surface spectrum.
    Consider a random vector $\mathbf{X} = ( X_1, \dots, X_n)$ chosen uniformly at random over $\leeball{\delta n, q}$. Hence, $\LW(\bfx) \leq \delta n$, where $\bfx$ denotes the realization of $\mathbf{X}$. It holds that
    \begin{align*}
        \log_2\left(\card{\leeball{\delta n, q}}\right) = H(\mathbf{X}),
    \end{align*}
    which implies, using again Theorem \ref{thm:chainRule}, that
    \begin{align*}
        \nu_{\delta n}^{(n)} = \frac{1}{n} H(\mathbf{X}) \leq H(X_1).
    \end{align*}
    Note that $H(X_1) \leq \log_2(q)$ for any parameter of $\delta \in \left[ 0, r\right]$. Hence, again since $B_\delta$ maximizes the entropy under the constraint $\expect(\LW(X_1)) \leq \delta$, we observe that $H(X_1) \leq H_\delta^+$ which yields the first statement of the lemma. To prove the latter statement it suffices to take the limit as $n$ tends to infinity.
\end{proof}

\subsection{Error Probability Bounds for the Constant Lee Weight Channel}
We consider an $[n, k]$ code $\code$ of cardinality $\card{\code} = q^k =: M$ over $\intmodq{q}$, and we focus on the {\constLC} where the additive error term is of fixed Lee weight $\delta n$. We denote by $P_B(\code)$ the block error probability of the code $\code$ under a given decoding rule.
\begin{theorem}[Random Coding Union Bound, ML Decoding]\label{thm:RCU_const_ML}
    Let $\code \subset (\intmodq{q})^n$ be a random code of rate $R_2$. The average \ac{ML} decoding error probability of $\code$ used to transmit over a {\constLC} satisfies
	\begin{align*}
	    \expect\left(\perror(\code)\right)< 2^{-n \left[\log_2 q - \sigma_{\delta n}^{(n)} - R_2\right]^+}.
	\end{align*}
\end{theorem}
\begin{proof}
    Consider first the pairwise error probability $\pep(\bfx, \bfy)$  for fixed $\bfx$ and $\bfy$, where $\bfx$ is the transmitted codeword, $\bfy$ is the channel output and $\widetilde{\mathbf{X}}$ is a random codeword distributed uniformly over $(\intmodq{q})^n$. By breaking ties always towards $\widetilde{\mathbf{X}}$, we can upper bound the pairwise error probability as
	\begin{align}
	\pep(\bfx,\bfy) &\leq \prob\left(P_{\mathbf{Y} \st \mathbf{X}} (\bfy \st \bfx) = P_{\mathbf{Y} \st \mathbf{X}} (\bfy \st \widetilde{\mathbf{X}}) \right) \\
          &= \prob\left(\LD(\bfy,\widetilde{\mathbf{X}})=\delta n\right)\\
          &= \frac{\card{\leesphere{\delta n, q}}}{q^n}.
	\end{align}
	The union bound on the block error probability is obtained by multiplying the result by $M-1$. By observing that the pairwise error probability does not depend on $\bfx, \bfy$, we get
	\begin{align}
		\expect\left(\perror(\code)\right) &\leq \min\left(1, (M-1)\pep(\bfx, \bfy)\right)\\
            &<\min\left(1, \, M \frac{\card{\leesphere{\delta n,q}}}{q^n}\right)\\
            &=2^{-n \left[\log_2 q - \sigma_{\delta n}^{(n)} - R_2\right]^+}.
	\end{align}
\end{proof}

Owing to Lemma \ref{lemma:bound_mu_H}, the bound can be loosened yielding the simple form described in the following corollary.

\begin{corollary}\label{cor:RCU_const_ML}
    The average \ac{ML} decoding error probability of a random code $\code \subset (\intmodq{q})^n$ of rate $R_2$ used to transmit over a {\constLC} satisfies
	\begin{align}
	\expect(\perror(\code))&< 2^{-n \left[\log_2 q - H_\delta - R_2\right]^+}\\
	&=2^{-n \left[\KLD{B_q}{\mathcal{U}(\intmodq{q})} - R_2\right]^+}.
	\end{align}
\end{corollary}

In terms of \ac{MD} decoding, the two results can be proven in a similar fashion, considering all codewords of distance up to $\delta n$, i.e., instead of working over the sphere $\leesphere{\delta n, q}$ only we extend to the ball $\leeball{\delta n, q}$. Then the \ac{MD} counterparts of Theorem \ref{thm:RCU_const_ML} and its consequence, Corollary \ref{cor:RCU_const_ML}, are given in the following two results.

\begin{theorem}[Random Coding Union Bound, MD decoding]\label{thm:RCU_const_MD}
    Let $\code \subset (\intmodq{q})^n$ be a random code of rate $R_2$. 
    The average \ac{MD} decoding error probability of $\code$ used to transmit over a {\constLC} satisfies
	\begin{align*}
	\expect(\perror(\code))< 2^{-n \left[\log_2 q - \nu_{\delta n}^{(n)} - R_2\right]^+}.
	\end{align*}
\end{theorem}
\begin{proof}
    Consider first the pairwise error probability under the assumption that $\bfx$ is the transmitted codeword, $\bfy$ is the channel output and $\widetilde{\mathbf{X}}$ is a random codeword distributed uniformly over $(\intmodq{q})^n$. By breaking ties always towards $\widetilde{\mathbf{X}}$, we have
	\begin{align}
		\pep(\bfx,\bfy) &\leq \prob\left(\LD(\bfy,\bfx)\geq \LD(\bfy,\widetilde{\mathbf{X}})\right)\\
		&= \prob(\LD(\bfy,\widetilde{\mathbf{X}})\leq \delta n)\\
		&= \frac{\card{\leeball{\delta n}}}{q^n}.
	\end{align}
	The union bound on the block error probability can be obtained by multiplying the result by $M-1$. By observing that the pairwise error probability does not depend on $\bfx, \bfy$, we get
	\begin{align}
		\expect(\perror(\code)) &\leq \min\left(1, (M-1)\pep(\bfx, \bfy)\right)\\
		&<\min\left(1,M \frac{\card{\leeball{\delta n}}}{q^n}\right)\\
		&=2^{-n \left[\log_2 q - \nu_{\delta n}^{(n)} - R_2\right]^+}.
	\end{align}
\end{proof}

Owing to Lemma \ref{lemma:bound_vol_H}, the bound can be loosened yielding the simple form described in Corollary \ref{cor:RCU_const_MD}.

\begin{corollary}\label{cor:RCU_const_MD}
    The average \ac{MD} decoding error probability of a random code $\code \subset (\intmodq{q})^n$ of rate $R_2$ used to transmit over a {\constLC} satisfies
    \begin{align*}
        \expect(\perror(\code))< 2^{-n \left[\log_2 q - H^+_\delta - R_2\right]^+}.
    \end{align*}
\end{corollary}

Figure \ref{fig:RCUB_constant} depicts the upper bounds based on Theorem \ref{thm:RCU_const_MD} and Corollary \ref{cor:RCU_const_MD} for \ac{MD} decoding, for $[500,250]$ codes over $\intmodq{7}$. The bound of Corollary \ref{cor:RCU_const_MD} is only slightly looser than the one provided by Theorem \ref{thm:RCU_const_MD}. A similar result holds for the bounds of Theorem \ref{thm:RCU_const_ML} and Corollary \ref{cor:RCU_const_ML}, under \ac{ML} decoding.

\begin{figure}[]
    \centering
    \begin{tikzpicture}[scale = 0.9, every node/.style={scale=0.85}]
    \begin{semilogyaxis}[
        width=0.7\columnwidth,
        height=0.5\columnwidth,
        grid=both,
        grid style={dotted,gray},
        legend cell align=left,
        legend style={font=\small},
        legend pos=north west,
        mark options={solid},
        mark size=3,
        xmin=0.25,
        xmax=0.4,
        ymin=1e-10,
        ymax=1e-1,
        xlabel={Normalized Lee weight $\delta$},
        ylabel={Block error probability}
        ]
        
        \addplot[teal!90!blue,line width = 1pt] table[x = delta, y=RCU] {RCUConst500q7.txt};\addlegendentry{RCU bound, Theorem \ref{thm:RCU_const_MD}};
        \addplot[purple!70!orange,line width =1pt,dashed] table[x = delta, y=RCUH] {RCUConst500q7.txt};\addlegendentry{RCU bound, Corollary \ref{cor:RCU_const_MD}};  

	\end{semilogyaxis}
\end{tikzpicture}
    \caption{Random coding union bounds under \ac{MD} decoding based on Theorem \ref{thm:RCU_const_MD} and Corollary \ref{cor:RCU_const_MD} for the parameters $n=500$ and $k = 250$ over $\intmodq{7}$.}
    \label{fig:RCUB_constant}
\end{figure}
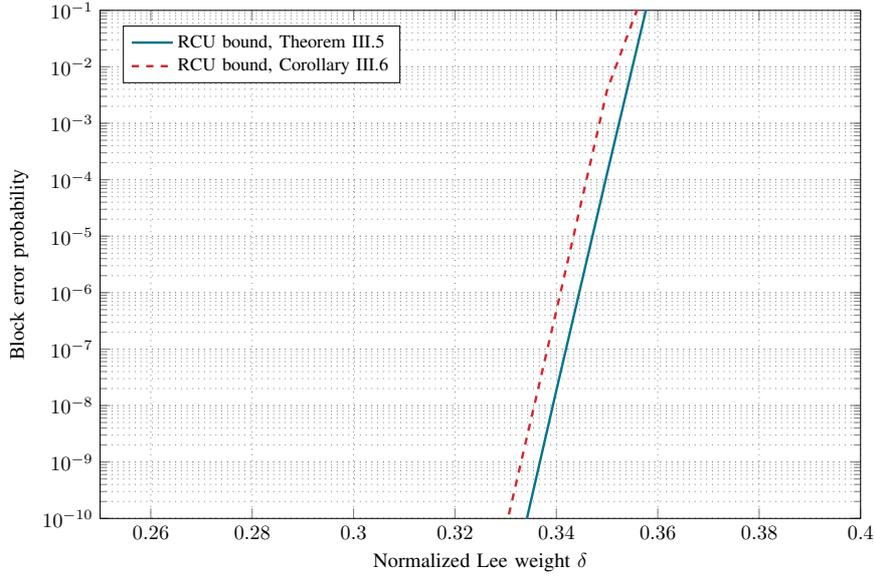

\subsection{Error Probability Bounds for the Memoryless Lee Channel}

We consider next a {\memLC} with expected normalized Lee weight of the error pattern $\delta$. We restrict the attention to the case  $\delta \leq \delta_q$. Recall that, in this regime, the \ac{ML} and the \ac{MD} decoding rules coincide since the \ac{ML} decoder gives a list of all codewords which are at distance $\delta n$ from the received word $\bfy$ and it outputs one of the codewords in this list randomly.

\begin{theorem}[Random Coding Union Bound]\label{thm:RCU_mmless_ML}
    Let $\code \subset (\intmodq{q})^n$ be a random code of rate $R_2$. The average \ac{ML}/\ac{MD} decoding error probability of $\code$ used to transmit over a {\memLC} with expected normalized Lee weight $\delta$ of the error pattern satisfies
    \begin{align}
        \expect\left(\perror(\code)\right)
            < 
        \expect\left(2^{-n \left[\log_2 q - \nu_{L}^{(n)} - R_2\right]^{+}}\right)
    \end{align}
     where the expectation is taken over the distribution of the Lee weight $L=\LW(\bfE)$.
\end{theorem}

\begin{proof}
    Let $\bfE \in (\intmodq{q})^n$ with $\LW(\bfE ) = L$, $\widetilde{\mathbf{X}}$ is a random codeword distributed uniformly over $(\intmodq{q})^n$, $\bfx$ the transmitted codeword and $\bfy$ the channel output. We estimate the pairwise error probability of $\bfx$ and $\bfy$ given that $\LW(\bfE) = L$. That is, by breaking ties always towards $\widetilde{\mathbf{X}}$, we get
    \begin{align}
        \pep_{L}(\bfx,\bfy) &\leq \prob\left(\LD(\bfy,\bfx)\geq \LD(\bfy,\widetilde{\mathbf{X}}) \st \LW(\bfE) =  L\right)\\
        &= \prob\left(\LD(\bfy,\widetilde{\mathbf{X}})\leq L \st \LW(\bfE) = L\right)\\
        &= \frac{\card{\leeball{L}}}{q^n}.
    \end{align}
    The union bound on the block error probability can be obtained by multiplying the result by $M-1$. By observing that the pairwise error probability does not depend on $\bfx, \bfy$, we get
    \begin{align}
        \expect(\perror(\code) \st  \LW(\bfE) = L) &\leq \min\left(1, (M-1)\pep_L(\bfx, \bfy)\right)\\
        &<\min\left(1,M \frac{\card{\leeball{L}}}{q^n}\right)\\
        &=2^{-n \left[\log_2 q - \nu_{L}^{(n)} - R_2\right]^+}.
    \end{align}
    Then, taking the expectation with respect to the Lee weight of $\bfE$ yields the desired result.
\end{proof}

A direct consequence using Lemma \ref{lemma:bound_vol_H} is captured in Corollary \ref{cor:RCU_mless}. Its proof follows similar to the constant Lee weight case.
\begin{corollary}\label{cor:RCU_mless}
   The average \ac{ML}/\ac{MD} decoding error probability of a random code $\code \subset (\intmodq{q})^n$ of rate $R_2$ used to transmit over a {\memLC} satisfies
    \[
        \expect (\perror(\code)) < \expect \left( 2^{-n \left[\log_2 q - H^{+}_{L/n}  - R_2\right]^{+}} \right)
    \]
    where the expectation is taken over the distribution of the Lee weight $L=\LW({\bfE})$.
\end{corollary}

Following the idea of \cite[Section 5.8]{Gal68}, we provide next a lower bound on the block error probability achievable by \emph{any} $[n,k]$ code over the {\memLC}.

\begin{theorem}[Sphere-Packing Bound]\label{thrm:SPB}
    The block error probability of any code $\code \subseteq (\intmodq{q})^n$ of rate $R_2$ over a {\memLC} is lower bounded as
    \begin{align*}
        \perror (\code) > \frac{1}{Z(\beta)^n}\sum_{d=d_0+1}^{rn} \card{\leesphere{d,q}}  \exp\left(-\beta d\right)+ \frac{1}{Z(\beta)^n}\left( \card{\leesphere{d_0,q}} - \xi\right)\exp\left(-\beta d_0\right)
    \end{align*}
    where $d_0$ and $\xi$ are chosen so that
    \begin{gather*}
        \sum_{d=0}^{d_0-1} \card{\leesphere{d,q}} + \xi = {2^{n(\log_2(q) - R_2)}} \quad \text{and} \\ 0<\xi\leq  \card{\leesphere{d_0, q}}.
    \end{gather*}
\end{theorem}
\begin{proof}
    The proof follows closely the analogous proof for the binary symmetric channel provided in \cite[Section 5.8]{Gal68}. Let $\code \subseteq (\intmodq{q})^n$ be a code of rate $R_2 = \frac{\log_2(M)}{n}$ and let $\bfx_1, \ldots , \bfx_M$ denote its codewords. Furthermore, we define for each codeword $\bfx_i$, the set $\mathcal{Y}_i$ of output sequences $\bfy$ such that $\bfy$ is decoded into $\bfx_i$. Within a decision region $\mathcal{Y}_i$, we let $A_{d, i}$ denote the number of sequences $\bfy \in \mathcal{Y}_i$ such that $\LD(\bfy, \bfx_i) = d$. The overall probability of correct decoding can then be computed as
    \begin{align*}
        P_{\mathsf{correct}} 
        &= \frac{1}{M} \sum_{i = 1}^M \sum_{\bfy \in \mathcal{Y}_i}  P_{\bfY \st \bfX}(\bfy \st \bfx_i) \\
        &= \frac{1}{M} \sum_{i = 1}^M \sum_{d = 0}^{n\floor{q/2}} A_{d, i} \frac{1}{Z(\beta)^n}\exp(-\beta d).
    \end{align*}
    This implies that the block error probability $\perror(\code)$ is computed as
    \begin{align}
        \perror(\code) 
        &= 1 - P_{\mathsf{correct}}\\
        &= \frac{1}{M} \sum_{i = 1}^M \sum_{d = 0}^{n\floor{q/2}} \left( \card{\leesphere{d,q}} - A_{d, i}\right) \frac{1}{Z(\beta)^n}\me^{-\beta d}.\label{eq:blockerrorprobability}
    \end{align}
    To obtain the desired lower bound, we minimize the expression on the right hand side of \eqref{eq:blockerrorprobability}. The expression is minimized subject to the constraints
    \begin{itemize}
        \item $A_{d, i} \leq \card{\leesphere{d, q}}$, for every $d \in \set{0, \ldots, n\floor{q/2}}$ and $i = \set{1, \ldots, M}$, and \\
        \item $\sum_{i = 0}^M \sum_{d = 0}^{n\floor{q/2}} A_{d, i} \leq q^n$.
    \end{itemize}
    It can be shown that the minimum is achieved for
    \begin{align}\label{eq:proofSPB_cases}
        A_{d, i} =
        \begin{cases}
            \card{\leesphere{d, q}} & 0 \leq d \leq d_0 - 1 \\
            0 & d_0 + 1 \leq d \leq n\floor{q/2}
        \end{cases},
    \end{align}
    where $d_0$ is chosen such that
    \begin{gather*}
        \sum_{d = 0}^{d_0 - 1}\card{\leesphere{d,q}} + \frac{1}{M} \sum_{i = 1}^M A_{d, i} = 2^{n(\log_2(q) - R_2)}, \quad \text{and}\\
        0 \leq \frac{1}{M} \sum_{i = 1}^M A_{d_0,i} \leq \card{\leesphere{d_0,q}} .
    \end{gather*}
    Denoting $\xi :=\frac{1}{M} \sum_{i = 1}^M A_{d_0,i}$ and substituting \eqref{eq:proofSPB_cases} in \eqref{eq:blockerrorprobability} yields the desired lower bound.
\end{proof}

Figure \ref{fig:SPBRCU_Gallager} depicts the random coding union bound of Corollary \ref{cor:RCU_mless} and the sphere-packing bound of Theorem \ref{thrm:SPB}, over a {\memLC}, for $[1024,512]$ codes over $\intmodq{7}$. The random coding union bound given in Corollary \ref{cor:RCU_mless} is tight with respect to the sphere-packing bound in Theorem \ref{thrm:SPB}. Hence, they provide an accurate benchmark to assess the performance achievable over the {\memLC}.

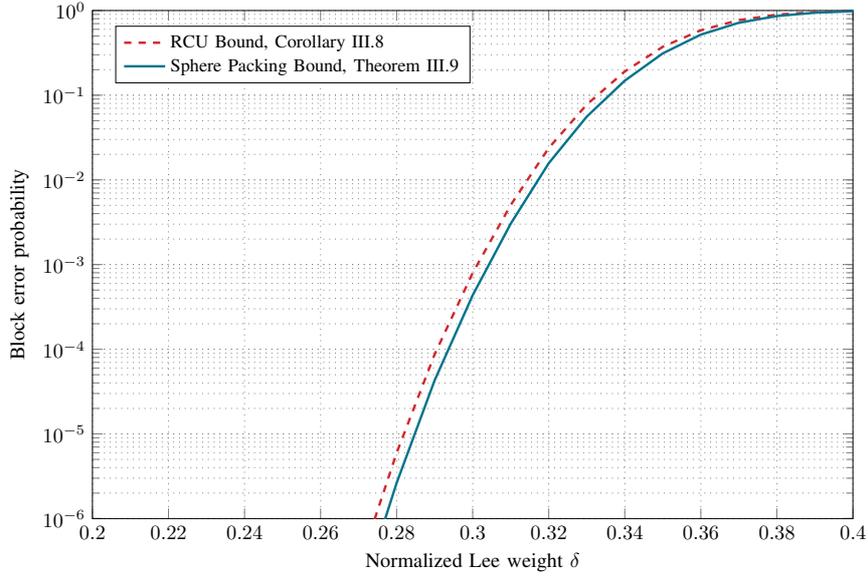
\begin{figure}[]
	\centering
	\begin{tikzpicture}[scale = 0.9, every node/.style={scale=0.85}]
    \begin{semilogyaxis}[
	width=0.7\columnwidth,
        height=0.5\columnwidth,
	grid=both,
	grid style={dotted,gray},
	legend cell align=left,
	legend style={font=\small},
	legend pos=north west,
	mark options={solid},
	mark size=3,
	xmin=0.2,
	xmax=0.4,
	ymin=1e-6,
	ymax=1,
	xlabel={Normalized Lee weight $\delta$},
	ylabel={Block error probability}
	]
	
	\addplot[purple!70!orange,line width = 1pt, dashed] table[x=delta,y=RCU] {RCUSPB.txt};\addlegendentry{RCU Bound, Corollary \ref{cor:RCU_mless}};
	\addplot[teal!90!blue,line width = 1pt] table[x=delta,y=SPB] {RCUSPB.txt};\addlegendentry{Sphere Packing Bound, Theorem \ref{thrm:SPB}};
	\end{semilogyaxis}
\end{tikzpicture}
	\caption{Random coding union (Corollary \ref{cor:RCU_mless}) and sphere-packing bounds (Theorem \ref{thrm:SPB}) for the parameters $n=1024$ and $k=512$ over $\intmodq{7}$.}\label{fig:SPBRCU_Gallager}
\end{figure}

\section{Lee Weight Spectrum of Regular LDPC Code Ensembles}\label{sec:LDPCdistances}

We now turn our attention to regular LDPC code ensembles over $\intmodq{q}$. We are going to derive the average Lee weight spectrum of a random code $\code \subseteq (\intmodq{q})^n$ from the $(d_v,d_c)$ regular \ac{LDPC} code ensemble. The result will be used in Section \ref{sec:LDPCperformance} to establish an upper bound on the block error probability under \ac{ML} decoding.

For each possible Lee weight $\ell \in \set{0, \ldots, n\floor{q/2}}$ we define the number of codewords of Lee weight $\ell$ as 
$$W_{\ell}^{(n)}(\code):= \card{\set{ \bfc \in \code \st \LW(\bfc) = \ell }}.$$
We are now interested in the number of entries of a certain Lee weight in a codeword, i.e., we are interested in the type in terms of the Lee weight of the codeword. For this we introduce the following definition.
\begin{definition}\label{def:type_codeword}
    For every codeword $\bfc \in \code$ we define its \textit{Lee type} to be the $(\floor{q/2}+1)$-tuple $\bm{\theta}_{\bfc} = \left(\theta_\bfc(0), \ldots, \theta_\bfc(\floor{q/2})\right)$ consisting of the relative fraction of occurrences of each possible Lee weight $\ell \in \set{0, \ldots,\floor{q/2}}$, i.e., 
    \begin{align*}
        \theta_{\bfc}(\ell) = \dfrac{1}{n}\card{ \set{ k = 1, \ldots , n \st \LW(c_k) =\ell} }.
    \end{align*}
\end{definition}
We denote the set of all Lee types over $(\intmodq{q})^n$ by $\mathcal{T}((\intmodq{q})^n)$. Then, we define the number of codewords in a code $\code \subseteq (\intmodq{q})^n$ of Lee type $\bm{\theta} \in \mathcal{T}((\intmodq{q})^n)$ as
$$A_{\bm{\theta}}^{(n)}(\code) := \card{\set{ \bfc \in \code \st \bm{\theta}_{\bfc} = \bm{\theta} }}.$$
Note that we can describe the Lee weight of a codeword $\bfc \in \code$ in terms of its Lee type as
\begin{align*}
    \LW(\bfc) = n \sum_{\ell = 1}^{\floor{q/2}} \ell \theta_{\bfc}(\ell).
\end{align*}
By abuse of notation, we will call this the \textit{Lee weight of the Lee type} $\bm{\theta}_{\bfc}$ and use the notation $\LW(\bm{\theta}_{\bfc})$.
Generally, for a Lee type $\bm{\theta} \in \mathcal{T}((\intmodq{q})^n)$ we define its Lee weight by $\LW(\bm{\theta}) := n \sum_{\ell = 1}^{\floor{q/2}} \ell \theta(\ell)$.
Thus, there is a natural relation between $W_{\ell}^{(n)}(\code)$ and $A_{\bm{\theta}}^{(n)}(\code)$. In fact, we have
\begin{align*}
    W_{\ell}^{(n)}(\code) = \sum_{\substack{\bm{\theta} \in \mathcal{T}\left( (\intmodq{q})^n\right) \\ \LW(\bm{\theta}) = \ell }} A_{\bm{\theta}}^{(n)}(\code).
\end{align*}

In the following, we consider a $(d_v,d_c)$-regular LDPC code $\code$ taken uniformly at random from an ensemble of $(d_v, d_c)$-regular LDPC codes over $\intmodq{q}$. Let $\bfH$ be a parity-check matrix of $\code$ where the nonzero entries of $\bfH$ lie in the set of units $\units{q}$. As $\code$ is a random regular LDPC code, the parity-check matrix $\bfH$ is a random matrix where each row has $d_c$ nonzero entries taken randomly among the unit elements and each column has $d_v$ of them. We consider a randomly chosen $\bfc \in (\intmodq{q})^n$ and denote its Lee type by $\bm{\theta}_{\bfc}$. Recall that $\bfc$ is a codeword if and only if $\bfc\bfH^\top = \bfzero$.

We now briefly discuss what it means for a codeword $\bfc$ of a random LDPC code to satisfy the check equations of a parity-check matrix $\bfH$.
Considering the Tanner graph of a code $\code$, given a codeword $\bfc$ we start by repeating each position $c_i$ exactly $d_{\vn}$ times over the edges connected to the $i$-th variable node. We denote the resulting vector by $\bfz' := (c_1, \ldots , c_1, \ldots, c_n, \ldots, c_n)$. Note that $\bfz'$ is of length $n d_{\vn}$ and is of Lee type $\bm{\theta}_{\bfz'} = \bm{\theta}_{\bfc}$. Let then $\bfu \in (\units{q})^{n d_{\vn}}$ be chosen uniformly at random, i.e., every entry $u_i$ is chosen uniformly at random among the units $\units{q}$. Finally, choosing a random permutation $\Pi$ we compute $\bfz := \Pi( \bfz' \odot \bfu )$. Now, $\bfc$ satisfies $\bfc\bfH^{\top} = \bfzero$ if and only if $\bfz$ satisfies the $m$ check equations induced by rows of $\bfH$.
Figure \ref{fig:random_PC_LDPC} below visualizes this procedure for a random $(d_{\vn}, d_{\cn})$-regular LDPC code.
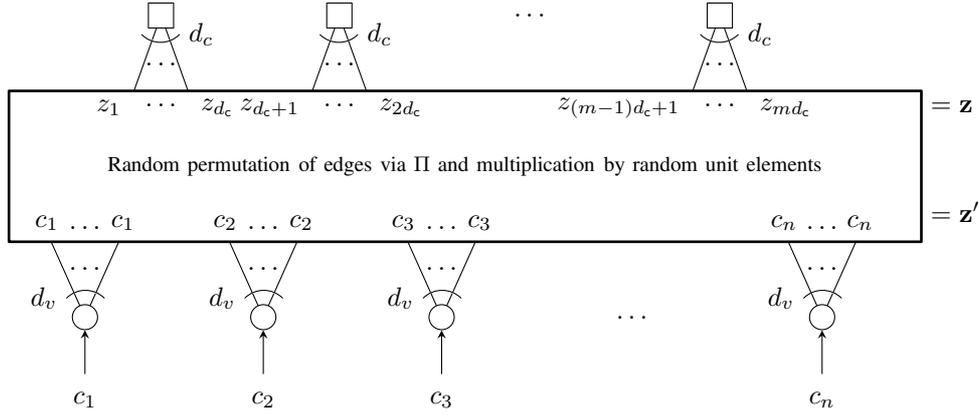
\begin{figure}
    \centering
    \begin{tikzpicture}[line cap=round,
        line join=round,
        >=triangle 45,
        x=1cm,
        y=1cm,
        thin,
        vnode/.style={draw,circle},
        vnot/.style={draw,circle, color = white},
        cnode/.style={draw,regular polygon,regular polygon sides=4},
        cnot/.style={draw,regular polygon,regular polygon sides=4, color = white}]

        \begin{scope}[xshift=0cm, yshift=0cm, start chain=going right, node distance=20mm]
            \foreach \i in {1, 2, 3}
              \node[vnode, on chain] (vn\i) [] {};
            \draw node[on chain] {$\ldots$};
            \node[vnode, on chain] (vn4) [] {};
        \end{scope}

        \begin{scope}
            \foreach \i in {1, 2, 3, 4}
              \node[above = 1, xshift = -0.5cm] at (vn\i) (vl\i) {};

            \foreach \i in {1, 2, 3, 4}
              \node[above = 0.5] at (vn\i) (vdot\i) {$\ldots$};

            \foreach \i in {1, 2, 3, 4}
              \node[above = 1, xshift = 0.5cm] at (vn\i) (vr\i) {};
        \end{scope}
        
        \draw (vn1) -- (vl1);
        \draw (vn1) -- (vr1);
        \draw (vn2) -- (vl2);
        \draw (vn2) -- (vr2);
        \draw (vn3) -- (vl3);
        \draw (vn3) -- (vr3);
        \draw (vn4) -- (vl4);
        \draw (vn4) -- (vr4);

        \begin{scope}
            \foreach \i in {1, 2, 3, 4}
              \node[circle,draw, minimum size = 20, white] (cvn\i) at (vn\i) {};

            \foreach \i in {1, 2, 3, 4}
              \draw (cvn\i.45) arc[start angle=45, end angle=135, radius=0.35] node[left] {$d_v$};
        \end{scope}        

        \begin{scope}
            \foreach \i in {1, 2, 3, 4}
                \node[below = 0.75] at (vn\i) (svn\i) {};

            \foreach \i in {1, 2, 3, 4}
                \draw[-stealth] (svn\i) -- (vn\i);

            \foreach \i in {1, 2, 3}
                \node[below] at (svn\i) (c\i) {$c_\i$};
            \node[below] at (svn4) (c4) {$c_n$};
        \end{scope}

        \draw[line width=1pt] (-1, 1) -- (-1, 3) -- (11,3) -- (11, 1) -- cycle;
        \draw (5, 2) node {\footnotesize Random permutation of edges via $\Pi$ and multiplication by random unit elements};
        
        \begin{scope}[xshift=1cm, yshift=4cm, start chain=going right, node distance=20mm]
            \foreach \i in {1, 2}
              \node[cnode, on chain] (cn\i) [] {};
            \draw node[on chain] {$\ldots$};
            \node[cnode, on chain] (cn3) [] {};
        \end{scope}

        \begin{scope}
            \foreach \i in {1, 2, 3}
              \node[below = 1, xshift = -0.4cm] at (cn\i) (cl\i) {};

            \foreach \i in {1, 2, 3}
              \node[below = 0.5] at (cn\i) (cdot\i) {$\ldots$};

            \foreach \i in {1, 2, 3}
              \node[below = 1, xshift = 0.4cm] at (cn\i) (cr\i) {};
        \end{scope}
        \draw (cn1) -- (cr1);
        \draw (cn1) -- (cl1);
        \draw (cn2) -- (cr2);
        \draw (cn2) -- (cl2);
        \draw (cn3) -- (cr3);
        \draw (cn3) -- (cl3);

        \begin{scope}
            \foreach \i in {1, 2, 3}
              \node[circle,draw, minimum size = 20, white] (ccn\i) at (cn\i) {};

            \foreach \i in {1, 2, 3}
              \draw (ccn\i.225) arc[start angle=225, end angle=315, radius=0.35] node[right] {$d_c$};
        \end{scope}   

    \begin{scope}
        \foreach \i in {1, 2, 3}
          \node[above = -0.1] at (vl\i) {$c_\i$};
        \node[above = -0.1] at (vl4) {$c_n$};

        \foreach \i in {1, 2, 3}
          \node[above = -0.1] at (vr\i) {$c_\i$};
          \node[above = -0.1] at (vr4) {$c_n$};

        \foreach \i in {1, 2, 3, 4}
          \node[above = 0.4] at (vdot\i) {$\ldots$};
    \end{scope}
    \node[above = 0.4, right] at (11, 1) {$= \bfz'$};

    \node[below = 0.1, left] at (cl1) {$z_1$};
    \node[below = 0.1, right] at (cr1) {$z_{d_{\cn}}$};
    \node[below = 0.1, left] at (cl2) {$z_{d_{\cn}+1}$};
    \node[below = 0.1, right] at (cr2) {$z_{2d_{\cn}}$};
    \node[below = 0.1, left] at (cl3) {$z_{(m -1)d_{\cn} + 1}$};
    \node[below = 0.1, right] at (cr3) {$z_{md_{\cn}}$};
    \begin{scope}
        \foreach \i in {1, 2, 3}
          \node[below = 0.4] at (cdot\i) {$\ldots$};
    \end{scope}
    \node[below = 0.2, right] at (11, 3) {$= \bfz$};
    \end{tikzpicture}
    
    \caption{Graphical representation of a random $(d_{\vn}, d_{\cn})$ LDPC code of length $n$.}
    \label{fig:random_PC_LDPC}
\end{figure}

Having Figure \ref{fig:random_PC_LDPC} in mind, we can say that the average Lee type enumerator of a random LDPC code is given by
\begin{align*}
    \overline{A}_{\bm{\theta}}^{(n)} = \binom{n}{n\bm{\theta}} \prob\left(\bfz \text{ \small satisfies the check equations} \st \bm{\theta}_{\bfc} = \bm{\theta}\right).
\end{align*}
We denote the Lee type of $\bfz$ by $\bm{\omega}_{\bfz}$ in order not to confuse it with the Lee type $\bm{\theta}_{\bfc}$. Note that $\bm{\omega}_{\bfz}$ highly depends on $\bm{\theta}_{\bfc}$. Further discussions and observations follow in Theorem \ref{thm:type2type_general}. For now, let $\mathcal{T}_{\bm{\theta}_{\bfc}}\left( (\intmodq{q})^{n d_{\vn}}\right)$ denote the set of all possible Lee types for a vector $\bfz$ resulting from the Lee type $\bm{\theta}_{\bfc}$. We will denote
\begin{align}
    f^{(n)}(\bm{\omega}  \st  \bm{\theta}) &:= \prob\left(  \bm{\omega}_{\bfz} = \bm{\omega}  \st  \bm{\theta}_{\bfc} = \bm{\theta} \right) \quad \text{and} \label{equ:type2type}\\
    a^{(n)}(\bm{\omega})  &:=  \prob\left(\bfz \text{ \small satisfies the check equations}  \st  \bm{\omega}_{\bfz} = \bm{\omega} \right). \label{equ:CNcodeword}
\end{align} 
Hence, we can further break down the conditional probability as 
\begin{align}\label{equ:averageTypeEnum}
    \overline{A}_{\bm{\theta}}^{(n)} = {\binom{n}{n\bm{\theta}}} \sum_{\bm{\omega} \in \mathcal{T}_{\bm{\theta}}\left( (\intmodq{q})^{n d_{\vn}}\right)} f^{(n)}(\bm{\omega} \st  \bm{\theta}) a^{(n)}(\bm{\omega})
\end{align}
In the following we elaborate more the two probabilities $f^{(n)}(\bm{\omega} \st  \bm{\theta})$ and $a^{(n)}(\bm{\omega})$.

\subsection{Transformation of the Lee Type}
We start by analyzing how the Lee type of $\bfc$ changes to the Lee type of the vector $\bfz$. More precisely, we now study the probability $f^{(n)}(\bm{\omega} \st \bm{\theta})$ that the vector $\bfz$ has a Lee type $\bm{\omega}_{\bfz} = \bm{\omega}$ given that the Lee type of the codeword $\bfc$ is $\bm{\theta}_{\bfc} = \bm{\theta}$. Recall that $\bfz\in (\intmodq{q})^{n d_{\vn}}$ is formed from $\bfc$ by repeating the entries $c_i$ each $d_{\vn}$ times and then multiplying each copy by a randomly chosen unit. This already implies that the fraction of zeros in $\bfz$ must be equal to the fraction of zeros in $\bfc$. Focusing on the nonzero entries of $\bfc$ we have to treat several cases separately, as the multiplication of a random nonzero element $x \in \intmodq{q}$ by a random unit $u \in \units{q}$ lies in different orbits.\\

Note that the group of units $\units{q}$ acts under multiplication on $\intmodq{q}$. For an element $a \in \intmodq{q}$ we define its orbit $\mathcal{O}_a$ as
\begin{align}\label{equ:orbit}
    \orbit{a} :=\set{ a \cdot u \st u \in \units{q} }.
\end{align}
Orbits induce an equivalence relation, i.e., two elements are equivalent if and only if they lie within the same orbit. Each orbit can be represented by a divisor of $q$. Let $\divisors{q}$ denote the set of divisors of $q$, i.e., 
\begin{align*}
    \divisors{q} := \set{ \ell \in \NN \; : \; \ell \st q }.
\end{align*}
Then the distinct orbits are given by $\orbit{d}$ for $d\in \divisors{q}$. 
\begin{example}
    We consider the integer residue ring $\intmodq{10}$. The set of divisors is given by
    \begin{align*}
        \mathbb{D}_{10} = \set{1, 2, 5, 10}.
    \end{align*}
    Hence, there are four orbits defined by the divisors of ten, namely,
    \begin{gather*}
        \orbit{1} = \units{10} = \set{1, 3, 7, 9}, \ \orbit{2} = \set{2, 4, 6, 8}, \\ \orbit{5} = \set{5} \ \text{and} \ \orbit{0} := \orbit{10} = \set{0}.
    \end{gather*}
\end{example}
By the definition of an orbit in \eqref{equ:orbit}, we observe that if an element $a$ lies in a given orbit $\orbit{d}$ then every multiple of $a$ by a unit element is in the same orbit. Hence, a codeword $\bfc$ and a vector $\bfz$ resulting from $\bfc$ have the same fraction of elements in an orbit $\orbit{d}$ for every divisor $d\in \divisors{q}$. For a codeword $\bfc$ with Lee type $\bm{\theta}_{\bfc}$ and for every $d \in \divisors{q}$ the fraction of elements in orbit $\orbit{d}$ is denoted as
\begin{align}\label{equ:fractionOrbits}
    \theta_{\bfc}(\orbit{d}) := \sum_{\substack{a \in \orbit{d} \\ a \leq \floor{q/2}}} \theta_{\bfc}(a).
\end{align}
The tuple of all such fractions is denoted by
\begin{align*}
    \bm{\theta}_{\bfc, \orbit{}} := \left( \theta_{\bfc}(\orbit{d_1}), \ldots , \theta_{\bfc}(\orbit{d_{\card{\divisors{q}}}}) \right).
\end{align*}

Regarding the Lee metric, we can prove that two elements of the same Lee weight are equivalent.
\begin{lemma}\label{lem:sameLee_sameOrbit}
    Elements of the same Lee weight in $\intmodq{q}$ lie in the same orbit, i.e., for every $a \in \intmodq{q}$ we have $\orbit{a} = \orbit{q-a}$.
\end{lemma}
\begin{proof}
    Let $a \in \intmodq{q}$. By symmetry of the Lee weight, $q-a$ is the only element having the same Lee weight as $a$. Let $b \in \orbit{q-a}$ be arbitray. By the definition of an orbit (see Equation \eqref{equ:orbit}), there exists a unit element $u \in \units{q}$ such that $b \equiv u(q-a) \equiv -ua \mod q$. Since $(-1)$ and $u$ are units, also $(-u)$ is a unit modulo $q$ and thus $b \in \orbit{a}$. Since $b$ was chosen arbitrarily, we have $\orbit{q-a} = \orbit{a}$.
\end{proof}

Lemma \ref{lem:sameLee_sameOrbit} indicates that we only have to consider elements up to $\floor{q/2}$. If $q$ is odd, then zero is the only element of Lee weight $0$. All other weights in this case are represented by two elements. If instead $q$ is even additionally the Lee weight $\floor{q/2}$ is represented only by one element, namely $\floor{q/2}$ itself. This fact is important when studying the number of configurations of a fixed Lee weight. Given the Lee type $\bm{\theta}_\bfx$ of a vector $\bfx$ we denote the fraction of Lee weights with only one representative element by
\begin{align*}
    \widehat{\bm{\theta}}_{\bfx} :=
    \begin{cases}
        1 - \theta_{\bfx}(0) & \text{if } q \text{ is odd},\\
        1 - \theta_{\bfx}(0) - \bm{\theta}_{\bfx}(\floor{q/2}). & \text{if } q \text{ is even}.
    \end{cases}
\end{align*}
We are then able to state the result on the expression for the probability $f^{(n)}(\bm{\omega}\st\bm{\theta})$ over $\intmodq{q}$.
\begin{theorem}\label{thm:type2type_general}
    Consider a random $\bfc$ of Lee type $\bm{\theta}_c$. Let $\bfz \in (\intmodq{q})^{n d_{\vn}}$ be the resulting vector when repeating the entries of $\bfc$ $d_{\vn}$ times and multiplying each position by a randomly chosen unit element. Furthermore, we denote by $\bm{\omega}_{\bfz}$ the Lee type of $\bfz$ and we define the set
    \begin{align*}
        \mathcal{T}_{\bm{\theta}_{\bfc}}\left( (\intmodq{q})^{n d_{\vn}}\right) := \{ \bm{\omega} \in \mathcal{T}((\intmodq{q})^{n d_{\vn}}) \st \;\omega(\orbit{d}) = \theta_{\bfc}(\orbit{d}) \, \forall d \in \divisors{q} \}.
    \end{align*} Let $\divisors{q} = \set{d_1, \ldots, d_r}$ be the set of divisors of $q$. If $\bm{\omega}_{\bfz} \in \mathcal{T}_{\bm{\theta}_{\bfc}}\left( (\intmodq{q})^{n d_{\vn}}\right)$, then we have
    \begin{align}\label{equ:type2type_prob}
        f^{(n)}(\bm{\omega}_{\bfz} \st \bm{\theta}_{\bfc}) =
            \frac{\binom{n d_{\vn} }{n d_{\vn}\bm{\omega}_{\bfz}} 2^{n d_{\vn}\widehat{\bm{\omega}_{\bfz}}}}{\binom{n d_{\vn}}{n d_{\vn}  \bm{\theta}_{\bfc, \orbit{}}} \prod_{d \in \divisors{q}} \card{\orbit{d}}^{n d_{\vn} \theta_{\bfc}(\orbit{d})} } 
    \end{align}
    If not, then the probability $f^{(n)}(\bm{\omega}_{\bfz} \st \bm{\theta}_{\bfc})$ is zero.
\end{theorem}
\begin{proof}
    Assume the Lee type $\bm{\theta}_{\bfc}$ of $\bfc$ is given by $\bm{\theta}$ and let the Lee type $\bm{\omega}_{\bfz}$ be equal to $\bm{\omega}$. By the above discussion, when multiplying an element $a$ of a given orbit $\orbit{d}$ with a randomly chosen unit $u \in \units{d}$, the product is still an element of $\orbit{d}$. In fact, $au$ can take each element of $\orbit{d}$ with the same probability. Therefore, $\bfz$ must have the same fraction of elements in orbit $\orbit{d}$ as the codeword $\bfc$ which also yields, that $f^{(n)}(\bm{\omega} \st \bm{\theta}) = 0$ if this is not fulfilled.

    Let us assume then that for every divisor $d$ of $q$ it holds that $\omega(\orbit{d}) = \theta(\orbit{d})$. The probability that $\bm{\omega}_{\bfz} = \bm{\omega}$ given that $\bm{\theta}_{\bfc} = \bm{\theta}$ is given by the number of vectors of length $n d_{\vn}$ over $\intmodq{q}$ of Lee type $\bm{\omega}$ divided by the total number of vectors of a Lee types fulfilling the constraint on the fraction of orbit elements. The number of configurations of vectors with Lee type $\bm{\omega}$ is given by the multinomial coefficient
    \begin{align*}
        \binom{n d_{\vn} }{n d_{\vn}\bm{\omega}} = \binom{n d_{\vn} }{n d_{\vn} \omega(0), \ldots , n d_{\vn} \omega(\floor{q/2})}.
    \end{align*}
    Since the Lee type gives rise only to the number of elements of a certain Lee weight, we must consider Lee weights reached by two different elements. We hence have to multiply the multinomial coefficient by a power of $2$ considering the two options for Lee weights admitting two representative elements given by $2^{n d_{\vn} \widehat{\bm{\omega}}}$. This yields us the numerator of the probability $f^{(n)}(\bm{\omega} \st \bm{\theta})$ and hence the number of vectors $\bfv \in (\intmodq{q})^{n d_{\vn}}$ of Lee type $\bm{\omega}$.

    We are now interested in finding the number of vectors $\bfv \in (\intmodq{q})^{n d_{\vn}}$ of Lee type $\bm{\omega}_{\bfv, }$ satisfying $\bm{\omega}_{\bfv, \orbit{}} (\orbit{d}) = \bm{\theta}_{\orbit{}}$. This number splits into two quantities: first, focusing only on the orbits, the number of constellation of the orbits, and second the number of choices in each orbit. The first quantity is again given by a multinomial coefficient regarding the fraction of elements in orbit $\orbit{d}$ for every $d \in \divisors{q}$ given in \eqref{equ:fractionOrbits}. To obtain the latter quantity we raise the cardinality of the orbit $\orbit{d}$ to the power of the number of positions with elements in that orbit. Combining the results yields the denominator and hence, the desired result on the probability $f^{(n)}(\bm{\omega} \st \bm{\theta})$.
\end{proof}

Note that if $q$ is a prime number, there are only two orbits; one containing only the zero element, and one corresponding to the set of units modulo $q$ (which are all nonzero elements). Then the expression in Theorem \ref{thm:type2type_general} simplifies to
\begin{align*}
    f^{(n)}(\bm{\omega} \st \bm{\theta}) =
    \begin{cases}
        \frac{2^{n d_{\vn}\widehat{\bm{\omega}}}}{(q-1)^{n d_{\vn} (1-\theta(0))}} & \text{if } \omega(0) = \theta(0) \\
        0 & \text{otherwise}.
    \end{cases}
\end{align*}

Furthermore, there is a closed form for the cardinalities of the orbits which allows for a simple implementation of the formula given in Theorem \ref{thm:type2type_general}.

\begin{lemma}
    Let $q$ be a positive integer and let $\divisors{q}$ be the set of divisors of $q$. Furthermore, let $\varphi(\cdot)$ denote the Euler totient function. Then, for every $d \in \divisors{q}$ the cardinality of its orbit is given by
    \begin{align*}
        \card{\orbit{d}} = \varphi(q/d).
    \end{align*}
\end{lemma}
\begin{proof}
    To compute the cardinality of the orbit $\orbit{d}$, we make use of Lagrange's theorem \cite[Theorem 1.8]{grove1983algebra} which yields
    \begin{align}\label{equ:stabilizerTHM}
        \card{\orbit{d}} = \frac{\card{\units{q}}}{\card{\mathrm{Stab}(d)}}
    \end{align}
    where $\mathrm{Stab}(d)  = \set{a \in \units{q} \st a d = d \mod q}$ is the stabilizer of $d$. Using that $d\st q$ and the definition of modular equality, we can rewrite $\mathrm{Stab}(d)$ as
    \begin{align}\label{equ:stabilizer}
        \mathrm{Stab}(d) &= \set{a \in \units{q} \st  q | (a-1)d} \\
        &= \set{a \in \units{q} \st  (q/d) | (a-1)}.
    \end{align}
    Since $\intmodq{q}$ is an additive group and since $d \st q$, $d$ is of order $\mathrm{ord}_{\intmodq{q}}(d) = q/d =: D$. Additionally, every element in $\orbit{d}$ is of the same order $D$. Consider now the canonical map $f:  \units{q}  \longrightarrow \units{D}$ defined by reducing the elements $x \in \units{q}$ modulo $D$. Since $D \st q$, the function $f$ is surjective and therefore, by the fundamental homomorphism theorem \cite[Theorem 1.11]{grove1983algebra}, we obtain 
    \begin{align*}
        \card{\ker(f)} = \frac{\card{\units{q}}}{\card{\units{D}}} = \frac{\varphi(q)}{\varphi(D)}.
    \end{align*}
    Finally, we note that
    \begin{align*}
        \ker(f) &= \set{a \in \units{q} \st a - 1= 0 \mod D} \\
        &= \set{a \in \units{q} \st  D | (a-1)}
    \end{align*}
    which corresponds exactly to the stabilizer in \eqref{equ:stabilizer}. Hence, using $\card{\mathrm{Stab}(d)} = \frac{\varphi(q)}{\varphi(q/d)}$ and $\card{\units{q}} = \varphi(q)$ in \eqref{equ:stabilizerTHM} yields the desired result.
\end{proof}

\begin{example}
    To illustrate \eqref{equ:type2type_prob} presented in Theorem \ref{thm:type2type_general} consider the following example over $\intmodq{16}$. Note that $\intmodq{16}$ consists of the following five orbits:
    \begin{gather*}
        \orbit{16} = \set{0}, \; \orbit{1} = \units{16}, \; \orbit{2} = \set{ 2, 6, 10, 14 },  \orbit{4} =\set{4, 12}\; \text{and} \; \orbit{8} = \set{8}.
    \end{gather*}
    Let $\code \subset (\intmodq{16})^2$ be a regular code with regular variable node degree $d_{\vn} = 2$. Let $\bfc \in \code$ be a codeword of Lee type $\bm{\theta}_\bfc = (0, 0, 1/2, 0, 1/2, 0, 0, 0, 0)$. Without loss of generality, we can assume that $\bfc = (2, 4)$. Following the procedure described by Figure \ref{fig:random_PC_LDPC} yields
    $$\bfz' = (2, 2, 4, 4).$$
    When multiplying each of the entries by a randomly chosen unit, we observe that $\bfz$ can be one of the following vectors (up to permutation and multiplication by $\pm 1$)
    \begin{align*}
        (2, 2, 4, 4),\; (2, 6, 4, 4),\; \text{and}\; (6, 6, 4, 4).
    \end{align*}
    Hence, the possible types for $\bfz$ are
    \begin{align*}
        &\bm{\omega}^{(1)} = (0, 0, 1/2, 0, 1/2, 0, 0, 0, 0), \\ &\bm{\omega}^{(2)} = (0, 0, 1/4, 0, 1/2, 0, 1/4, 0, 0) \;\text{and}\\ &\bm{\omega}^{(3)} = (0, 0, 0, 0, 1/2, 0, 1/2, 0, 0).
    \end{align*}
    The number of permutations for each case is given by the multinomial coefficient with respect to the Lee type $\bm{\omega}^{(i)}$. For instance, the vector $(2, 2, 4, 4)$ admits $6$ permutations, i.e.,
    \begin{align*}
         \binom{n d_{\vn} }{n d_{\vn} \bm{\omega}^{(1)}(0), \ldots , n d_{\vn} \bm{\omega}^{(1)}(8)} 
        =  \binom{2\cdot 2}{2\cdot2\cdot(1/2), 2\cdot2\cdot(1/2)}
        = \frac{4!}{2!2!} = 6.
    \end{align*}
    Since the Lee type focuses on the Lee weight only and since every nonzero entry different from $\floor{q/2}$ admits two representatives, we have two possible entries for each position. In the case of Lee type $\bm{\omega}^{(1)}$ we would hence have $6\cdot 16 = 96$ possible vectors of that type. Similarly, we have $96$ vectors of Lee type $\bm{\omega}^{(3)}$ and $192$ vectors of Lee type $\bm{\omega}^{(2)}$. This yields a total of $384$ vectors. Note that this indeed coincides with
    \begin{align*}
        \binom{n d_{\vn}}{n d_{\vn}  \theta_{\bfc}(\orbit{1}), \ldots , n d_{\vn} \theta_{\bfc}(\orbit{16})} \prod_{d \in \divisors{q}} \card{\orbit{d}}^{n d_{\vn} \theta_{\bfc}(\orbit{d})}= \binom{4}{2}\card{\orbit{2}}^2\card{\orbit{4}}^2 
        = 384.
    \end{align*}
    Thus, the probability that $\bfz$ has Lee type $\bm{\omega}^{(1)}$ given that the Lee type of the codeword $\bfc$ is $\bm{\theta}_\bfc$ is $f^{(n)}(\bm{\omega}^{(1)} \st \bm{\theta}_\bfc) = \frac{96}{384} = \frac{1}{4}$.
\end{example}

Consequently to Theorem \ref{thm:type2type_general} we determine the asymptotic growth rate of $f^{(n)}(\bm{\omega} \st \bm{\theta})$ in Corollary \ref{cor:generating_probF_asympt}
\begin{corollary}\label{cor:generating_probF_asympt}
    Let $\bfz \in (\intmodq{q})^{n d_{\vn}}$ be the vector resulting from a vector $\bfc \in (\intmodq{q})^n$ of Lee type $\bm{\theta}$ after repetition and permutation. Then we obtain the following asymptotic expression for the probability that $\bfz$ is of Lee type $\bm{\omega}$
    \begin{align*}
        \phi(\bm{\omega} \st \bm{\theta}) :&= \lim_{n \tendsto \infty} \frac{1}{n}\log_2( f^{(n)}(\bm{\omega} \st \bm{\theta})) \\
        &= d_{\vn}  \Big(  H(\bm{\omega}) + \widehat{\bm{\omega}} - H(\bm{\theta}_{\orbit{}}) - \sum_{d \in \divisors{q}} \theta(\orbit{d})\log_2(\card{\orbit{d}}) \Big) .
    \end{align*}
\end{corollary}
\begin{proof}
    The proof follows by taking the limit of each summand.
\end{proof}
Moreover, Lemma \ref{lem:uniform_fn} shows us an even stronger form of convergence.
\begin{lemma}\label{lem:uniform_fn}
    Given a random regular $(d_{\vn}, d_{\cn})$ LDPC code over $\intmodq{q}$ and a Lee type $\bm{\theta} \in \mathcal{T}((\intmodq{q})^n)$. Consider the sequence $f^{(n)}(\bm{\omega} \st \bm{\theta})$ of probabilities, defined in \eqref{equ:type2type_prob}, with $\bm{\omega} \in \mathcal{T}_{\bm{\theta}}\left( (\intmodq{q})^{n d_{\vn}}\right)$. Then the sequence $\left(\frac{1}{n}\log_2  (f^{(n)}(\bm{\omega}\st\bm{\theta})) \right)_{n \in \mathbb{N}}$ is uniformly convergent to $\phi(\bm{\omega}\st\bm{\theta})$ as $n \tendsto \infty$.
\end{lemma}
\begin{proof}
    We have to show that for every $\varepsilon > 0$ there is a natural number $n_{\varepsilon} \in \mathbb{N}$ such that for all $n \geq n_{\varepsilon}$ it holds
    \begin{align*}
        \card{ \frac{1}{n}\log_2 ( f^{(n)}(\bm{\omega}\st\bm{\theta}))  - \phi(\bm{\omega}\st\bm{\theta}) } < \varepsilon .
    \end{align*}
    Applying Theorem \ref{thm:type2type_general} and Corollary \ref{cor:generating_probF_asympt}, and by using the triangle inequality, we get
    \begin{align*}
        \Bigg| \frac{1}{n} \log_2 ( f^{(n)}(\bm{\omega}\st\bm{\theta})) - \phi(\bm{\omega}\st\bm{\theta})\Bigg| 
            &=
        \Bigg|  \frac{1}{n} \log_2\left( \binom{n d_{\vn}}{n d_{\vn} \bm{\omega}}\right) -  d_{\vn} H(\bm{\omega}) - \frac{1}{n} \log_2\left( \binom{n d_{\vn}}{n d_{\vn} \bm{\theta}_{\orbit{}}}\right) + d_{\vn} H(\bm{\theta}_{\orbit{}}) \Bigg| \\
            &\leq 
        \Bigg|  \frac{1}{n} \log_2\left( \binom{n d_{\vn}}{n d_{\vn} \bm{\omega}}\right) -  d_{\vn} H(\bm{\omega})\Bigg| + \Bigg|  d_{\vn} H(\bm{\theta}_{\orbit{}}) - \frac{1}{n} \log_2\left( \binom{n d_{\vn}}{n d_{\vn} \bm{\theta}_{\orbit{}}}\right) \Bigg|.
    \end{align*}
    Let us focus now on $\Big|  \frac{1}{n} \log_2\left( \binom{n d_{\vn}}{n d_{\vn} \bm{\omega}}\right) -  d_{\vn} H(\bm{\omega})\Big|$. Recall from \eqref{equ:bounds_multinomial} that we have the following bounds on $\binom{n d_{\vn}}{n d_{\vn} \bm{\omega}}$,
    \begin{align*}
        \frac{1}{(n d_{\vn} +1)^{\floor{q/2} + 1}} 2^{n d_{\vn}H(\bm{\omega})} \leq \binom{n d_{\vn}}{n d_{\vn} \bm{\omega}} \leq 2^{n d_{\vn}H(\bm{\omega})}.
    \end{align*}
    Hence, if $\frac{1}{n} \log_2\left( \binom{n d_{\vn}}{n d_{\vn} \bm{\omega}}\right) > d_{\vn} H(\bm{\omega})$, we get
    \begin{align*}
        \Big|  \frac{1}{n} \log_2\Bigg( \binom{n d_{\vn}}{n d_{\vn} \bm{\omega}}\Bigg) -  d_{\vn} H(\bm{\omega})\Big| = 0.
    \end{align*}
    On the other hand, we obtain
    \begin{align*}
        \Big|  \frac{1}{n} \log_2\Big( \binom{n d_{\vn}}{n d_{\vn} \bm{\omega}}\Big) &-  d_{\vn} H(\bm{\omega})\Big|
            \leq 
        (\floor{q/2} + 1)\frac{1}{n}\log_2{(n d_{\vn}+1)}.
    \end{align*}
    By l'Hôpital's rule this converges to zero as $n $ tends to $ \infty$.

    Note that the same argument holds for $\Big| d_{\vn} H(\bm{\theta}_{\orbit{}}) - \frac{1}{n} \log_2\left( \binom{n d_{\vn}}{n d_{\vn} \bm{\theta}_{\orbit{}}}\right) \Big|$ and thus the result follows.
\end{proof}

\subsection{Valid Check Node Assignments}
We now discuss the probability $a^{(n)}(\bm{\omega})$ given in \eqref{equ:CNcodeword}. We make use of generating functions to describe the situation at one check node and then extend the generating function to $m$ check nodes. In the following let $w$ denote the Lee weight decomposition of a vector $\bfx \in (\intmodq{q})^n$. That is, for every $i = 0, \ldots , \floor{q/2}$,
\begin{align*}
    w_i = \card{\set{k = 1, \ldots , n \st \LW(x_k) = i}}.
\end{align*}

Furthermore, recall from Equation \eqref{equ:type2type_prob} in Theorem \ref{thm:type2type_general} that given a Lee type $\bm{\theta}$ of $\bfc$, the Lee type $\bm{\omega}$ of a valid check node assignment has to show the same orbit distribution. Hence, there is a restricted choice. Let us denote the set of possible check node types resulting from $\bm{\theta}$ by $\mathcal{T}_{\bm{\theta}}\left( (\intmodq{q})^{n d_{\vn}}\right)$, i.e., it is the set
\begin{align*}
    \set{ \bm{\omega} \in \mathcal{T}((\intmodq{q})^{n d_{\vn}}) \st \omega(\orbit{d}) = \theta(\orbit{d}) \, \forall d \in \divisors{q} }.
\end{align*}
In the following, for a given polynomial $p(x)$, we denote by $\mathrm{coeff}(p(x), x^i)$ the coefficient of $x^i$ in $p(x)$. 
\begin{theorem}\label{thm:generating_probA}
    Consider a vector $\bfz \in (\intmodq{q})^{n d_{\vn}}$ of Lee type $\bm{\omega}$ and weight decomposition $w$. Furthermore, consider a random regular LDPC code of variable degree $d_{\vn}$ and check node degree $d_{\cn}$. Then, the probability that $\bfz$ fulfills the check node equations is given by
    \begin{align*}
        a^{(n)}(\bm{\omega}) = \frac{\mathrm{coeff}(G(\bft), \bft^{\bm{\omega} n d_{\vn}})}{\binom{n d_{\vn}}{n d_{\vn} \bm{\omega}}},
    \end{align*}
    where
    \begin{align}\label{equ:genFunc_mCNs}
        G(\bft) = \frac{1}{q^m} \left[\sum_{\substack{\bfz_i\in (\intmodq{q})^{d_{\cn}}\\  
        d_{\cn}\bm{\omega}_{\bfz_i} = w} } \sum_{s = 0}^{q-1} \prod_{k = 1}^{d_{\cn}} \me^{\frac{2 \pi i}{q} s z_k} t_1^{n d_{\vn} \bm{\omega}(1)} \ldots t_{\floor{q/2}}^{n d_{\vn} \omega(\floor{q/2})} \right]^{m} .
    \end{align}
\end{theorem}
\begin{proof}
    Recall that $a^{(n)}(\bm{\omega})$ describes the probability of $\bfz \in (\intmodq{q})^{n d_{\vn}}$ satisfying the check node equations and being of a given Lee type $\bm{\omega}$. Furthermore, we have $m$ check nodes each of degree $d_{\cn}$. Hence, we can split $\bfz$ into $m$ parts $\bfz_1, \ldots , \bfz_{m}$ each one corresponding check node $\cn_1, \ldots , \cn_{m}$, respectively.
    
We focus now on one check node only and describe a generating function for the number of $\bfz_i$'s satisfying the check node equations of check node $\cn_i$ and having Lee weight decomposition given by $w = (w_0, \ldots, w_{\floor{q/2}})$. We turn our attention at this point only to the nonzero elements and note that $w_0 = d_{\cn} - \sum_{i = 1}^{\floor{q/2}} w_i$.
    In that sense, let us define
    \begin{align*}
        g_{(w_1, \ldots , w_{\floor{q/2}})} := \Big| \Big\lbrace \bfz_i \in (\intmodq{q})^{d_{\cn}} \big|  \bfz_i  \text{ satisfies the check-equations, and }  \card{\set{j = 1, \ldots , d_{\cn} \st \LW(z_{i_j}) = k}} = w_k \Big\rbrace \Big|.
    \end{align*}
    We can describe this quantity summing over all $d_{\cn}$-tuples that sum up to zero using an indicator function. Indeed,
    \begin{align*}
        g_{(w_1, \ldots , w_{\floor{q/2}})} = \sum_{\substack{\bfz_i\in (\intmodq{q})^{d_{\cn}}\\ d_{\cn} \bm{\omega}_{\bfz_i} = w} } \mathbb{1}\left( \sum_{k = 1}^{d_{\cn}} z_k  = 0\right).
    \end{align*}
    Applying the inversion formula for the discrete Fourier transform over $\intmodq{q}$ yields
    \begin{align}\label{equ:generate_cn_function}
        g_{(w_1, \ldots , w_{\floor{q/2}})} 
        = \sum_{\substack{\bfz_i\in (\intmodq{q})^{d_{\cn}}\\ d_{\cn} \bm{\omega}_{\bfz_i} = w} } \frac{1}{q}\sum_{\chi \text{ character }} \chi \left(\sum_{k = 1}^{d_{\cn}} z_k \right).
    \end{align}
    Over the finite abelian group $\intmodq{q}$ there are $q$ characters $\chi_0, \ldots, \chi_{q-1}$ defined by $\chi_k(a) := \me^{\frac{2 \pi i}{q} k a}$ for each element $a \in \intmodq{q}$. Hence, we can rewrite \eqref{equ:generate_cn_function} as
    \begin{align}\label{equ:generate_cn_function_closed}
        g_{(w_1, \ldots , w_{\floor{q/2}})} 
        = \frac{1}{q}\sum_{\substack{\bfz_i\in (\intmodq{q})^{d_{\cn}}\\ d_{\cn} \bm{\omega}_{\bfz_i} = w} } \sum_{s = 0}^{q-1} \me^{\frac{2 \pi i}{q} s \sum_{k = 1}^{d_{\cn}}z_k} .
    \end{align}
    We then define the generating function $g(\bft)$ by
    \begin{align}
        g(\bft) := \sum_{\substack{w \text{ composition}\\ \text{of } d_{\cn}}} g_{(w_1, \ldots , w_{\floor{q/2}})} t_1^{w_1}\ldots t_{\floor{q/2}}^{w_{\floor{q/2}}}.
    \end{align}
    
    To obtain a similar expression for a configuration regarding all the check nodes, we take the $m$-fold convolution of $g_{(w_1, \ldots , w_{\floor{q/2}})}$, i.e.,
    \begin{align*}
        G_{(w_1, \ldots , w_{\floor{q/2}})} :=  g_{(w_1, \ldots , w_{\floor{q/2}})} \circledast \ldots \circledast g_{(w_1, \ldots , w_{\floor{q/2}})}.
    \end{align*}
    
    Hence, the corresponding generating function for $m$ check nodes is 
    \begin{align*}
        G(\bft) := \sum_{\substack{w \text{ composition}\\ \text{of } m d_{\cn}}} g_{(w_1, \ldots , w_{\floor{q/2}})}t_1^{w_1}\ldots t_{\floor{q/2}}^{w_{\floor{q/2}}} = g(\bft)^m.
    \end{align*}
    Let $\bm{\omega}$ denote the Lee type of the decomposition $w$, i.e., $n d_{\vn} \omega(i) = w_i$ for every $i \in \set{0, \ldots ,\floor{q/2}}$. The number of configurations of given Lee type $\bm{\omega}$ is then the coefficient of the polynomial $G(\bft)$ at $\bft^{n d_{\vn} \bm{\omega}} = t_1^{w_1}\ldots t_{\floor{q/2}}^{w_{\floor{q/2}}}$. Finally, the probability $a^{(n)}(\bm{\omega})$ is obtained by dividing the $n d_{\vn} \bm{\omega}$-th coefficient of $G(\bft)$ by all the possible permutations of a vector $\bfx \in (\intmodq{q})^{n d_{\vn}}$ of Lee type $\bm{\omega}$, which is given by the multinomial coefficient $\binom{n d_{\vn}}{n d_{\vn} \bm{\omega}}$.
\end{proof}

At this point, to simplify the understanding we would like to discuss the expression in \eqref{equ:generate_cn_function_closed} with an example.
\begin{example}
    Assume the check node degree is $d_{\cn} = 2$ and that the underlying integer ring is $\intmodq{5}$. Let us furthermore assume that the Lee weight decomposition of a tuple $\bfz_i$ at a check node is $w = (0, 2, 0)$. This means that $\bfz_i$ is one of the following tuples
    \begin{align*}
        (1, 1),\quad (1, 4), \quad (4, 1), \quad \text{or} \quad (4, 4).
    \end{align*}
    Since only $(1, 4)$ and $(4, 1)$ satisfy the check equation (i.e. sum up to zero modulo five), the enumerator $g_{(0, 2, 0)}$ should equal two. In fact, the exponential expression in Equation \eqref{equ:generate_cn_function_closed} equals $1$ for all tuples satisfying the check equation. For those not satisfying the check equation the sum of exponentials is the sum of $n$-th roots of unity (in our case $n = 5$) and is hence equal to zero.
\end{example} 

Let us now focus on the asymptotic growth rate of $a^{(n)}$ which we define as
\begin{align*}
    \alpha(\bm{\omega}) := \lim_{n \tendsto \infty} \frac{1}{n}\log_2 (a^{(n)}(\bm{\omega})).
\end{align*}
A direct consequence of taking the logarithm and the limit of the sequence $a^{(n)}$ is captured in Corollary \ref{cor:generating_probA_asympt}.
\begin{corollary}\label{cor:generating_probA_asympt}
    Let $\bfz \in (\intmodq{q})^{n d_{\vn}}$ satisfy the $m$ check equations and denote by $\bm{\omega}$ its Lee type. Then we obtain the following asymptotic expression for the probability $a^{(n)}(\bm{\omega})$.
    \begin{align*}
        \alpha(\bm{\omega}) = -d_{\vn} H(\bm{\omega}) + (1-R) \inf_{\bft\succ 0} \log_2\left(\frac{g(\bft)}{\bft^{\bm{\omega} n d_{\vn}}} \right),
    \end{align*}
    where $\bft \succ 0$ means that not every entry of $\bft = (t_1, \ldots , t_{\floor{q/2}})$ is equal to zero.
\end{corollary}
Taking the infimum over all possibilities of $\bft = (t_1, \ldots , t_{\floor{q/2}})$ is impractical. We will use the asymptotic Hayman method for multivariate polynomials (see \cite{hayman1956generalisation, wilf2005generatingfunctionology, di2004asymptotic}) to establish $\lim_{n \tendsto \infty} 1/n \log_2 \left( \textrm{coeff}(G(\bft), \bft^{\bm{\omega} n d_{\vn}} \right)$.
\begin{lemma}[Hayman Formula]\label{lem:hayman}
    Let $\bfx = (x_1, \ldots, x_d) \in \mathbb{R}^d$ and let $p(\bfx)$ be a multivariate polynomial with $p(\bfzero) \neq 0$. Furthermore, let $\bm{\beta} = (\beta_1, \ldots , \beta_d)$ such that $0 \leq \beta_i \leq 1$ and $\beta_i n \in \mathbb{N}$ for all $i = 1, \ldots , d$. Assume that $\bfx^\star = (x_1^\star, \ldots , x_d^\star)$ is the unique positive real solution to the system of equations given by
    \begin{align*}
        x_1 \frac{\partial p(\bfx)}{\partial x_1} = \beta_1 p(\bfx), \; \ldots,\; x_d \frac{\partial p(\bfx)}{\partial x_d} = \beta_d p(\bfx).
    \end{align*}
    Then, as $n \tendsto \infty$, it holds
    \begin{align*}
        \lim_{n \longrightarrow\infty} \frac{1}{n} \ln \left( \mathrm{coeff}\left( (p(\bfz))^n, \bfz^{n \mathbf{\beta}} \right) \right)
            =
        \left( \ln(p(\bfx)) - \sum_{i = 1}^d \beta_i \ln(x_i) \right).
    \end{align*}
\end{lemma} 
In our case, we have that
\begin{align*}
    \lim_{n \longrightarrow\infty} \frac{1}{n} \ln\left( \mathrm{coeff}\left((g(\bft)^{1/d_{\cn}})^{n d_{\vn}}, \bft^{\mathbf{\bm{\omega}}n d_{\vn}}\right) \right) 
        =
    d_{\vn} \lim_{n' \longrightarrow\infty} \frac{1}{n'} \ln\left( \mathrm{coeff}\left((g(\bft)^{1/d_{\cn}})^{n'}, \bft^{\mathbf{\bm{\omega}}n'}\right) \right).
\end{align*}
Hence, Corollary \ref{cor:asympt_CNcodewords} is a direct consequence of Hayman's Formula.
\begin{corollary}\label{cor:asympt_CNcodewords}
    Let $\bm{\omega} = (\omega(0), \ldots, \omega(\floor{q/2})) \in (0, 1)^{\floor{q/2} + 1}$ such that $\omega(i) n d_{\vn} \in \mathbb{N}$ for every $i = 1, \ldots , d$. Then
    \begin{align*}
        \alpha(\bm{\omega}) = d_{\vn}\left( H(\bm{\omega}) + \log_2\left(g(\bft^\star)^{1/d_{\cn}}\right) - \sum_{i = 1}^{\floor{q/2}} \omega(i) \log_2(t_i^\star) \right)
    \end{align*}
    where $\bft^\star = (t_1^\star, \ldots, t_{\floor{q/2}}^\star)$ is the unique positive real solution to the equations
    \begin{align*}
        t_i \frac{\partial g(\bft)^{1/d_{\cn}}}{\partial t_i} = \omega(i) g(\bft)^{1/d_{\cn}}, \quad i = 1, \ldots , \floor{q/2}.
    \end{align*}
\end{corollary}
In his paper, Hayman gave an explicit expression for the coefficient of an admissible function (see \cite[p. 69]{hayman1956generalisation}). With this, it easily follows that the sequence of functions $\left(\frac{1}{n} \log_2(a^{(n)})\right)_{n \in \mathbb{N}}$ is uniformly convergent.

\subsection{Asymptotic Growth Rate}
Having determined the two probabilities defined in Equations \eqref{equ:type2type} and \eqref{equ:CNcodeword}, respectively, the expression for the average Lee type enumerator $\overline{A}_{\bm{\theta}}^{(n)}$ follows immediately. We can then deduce immediately the asymptotics of the average Lee type enumerator and average weight enumerator, respectively.

\medskip

\begin{corollary}\label{cor:asymptotic_weightEnum}
    Let $\code$ be a random $(d_{\vn}, d_{\cn})$-regular LDPC code of length $n$ over $\intmodq{q}$. Let us denote by $\mathcal{A}(\bm{\theta}) := \lim_{n\tendsto \infty} \frac{1}{n} \log_2\overline{A}_{\bm{\theta}}^{(n)}$ and $\mathcal{W}(\ell) := \lim_{n\tendsto \infty} \frac{1}{n} \log_2\overline{W}_{\ell}^{(n)}$ the spectral growth rate of the average Lee type enumerator and weight enumerator, respectively. Then
    \begin{align}
        \mathcal{A}(\bm{\theta}) 
        \leq H(\bm{\theta}) + \sup_{\bm{\omega} \in \mathcal{T}_{\bm{\theta}}\left( (\intmodq{q})^{n d_{\vn}}\right)}\left( \phi(\bm{\omega}\st\bm{\theta}) + \alpha(\bm{\omega}) \right)
    \end{align}
    and, in particular,
    \begin{align}
        \mathcal{W}(\ell) \leq \sup_{\substack{\bm{\theta} \in \mathcal{T}((\intmodq{q})^n) \\ \LW(\bm{\theta}) = \ell}} \mathcal{A}(\bm{\theta}).\label{equ:lim_weightSpec}
    \end{align}
\end{corollary}
\begin{proof}
    From Equation \eqref{equ:averageTypeEnum} we observe that 
    \[
    \bar{A}_{\bm{\theta}}^{(n)} = \binom{n}{n\bm{\theta}} \sum_{\bm{\omega} \in \mathcal{T}_{\bm{\theta}}\left( (\intmodq{q})^{n d_{\vn}}\right)} f^{(n)}(\bm{\omega} \st \bm{\theta}) a^{(n)}(\bm{\omega})
    \]
    and hence
    \begin{align*}
        \mathcal{A}&(\bm{\theta}) = H(\bm{\theta}) + \lim_{n \tendsto \infty} \frac{1}{n} \log_2\left(\sum_{\bm{\omega} \in \mathcal{T}_{\bm{\theta}}\left( (\intmodq{q})^{n d_{\vn}}\right)}  f^{(n)}(\bm{\omega} \st  \bm{\theta})  a^{(n)}(\bm{\omega})\right).
    \end{align*}
    Furthermore, we can write
        \begin{align*}
            \frac{1}{n} \log_2\Bigg(\sum_{\bm{\omega} \in \mathcal{T}_{\bm{\theta}}\left( (\intmodq{q})^{n d_{\vn}}\right)}  f^{(n)}(\bm{\omega} \st \bm{\theta})  a^{(n)}(\bm{\omega})\Bigg)
            &\leq 
            \frac{1}{n} \log_2\Bigg(\sup_{\bm{\omega} \in \mathcal{T}_{\bm{\theta}}\left( (\intmodq{q})^{n d_{\vn}}\right)} \left[ f^{(n)}(\bm{\omega} \st \bm{\theta})  a^{(n)}(\bm{\omega})\right]\card{\mathcal{T}_{\bm{\theta}}\left( (\intmodq{q})^{n d_{\vn}}\right)}\Bigg) \\
            &\overset{(a)}{=}
            \sup_{\substack{\bm{\omega} \in \\ \mathcal{T}_{\bm{\theta}}\left( (\intmodq{q})^{n d_{\vn}}\right)}}\left[ \frac{1}{n} \log_2 \left(  f^{(n)}(\bm{\omega} \st \bm{\theta})  a^{(n)}(\bm{\omega})\right)\right]\\
            &=
            \sup_{\substack{\bm{\omega} \in \\ \mathcal{T}_{\bm{\theta}}\left( (\intmodq{q})^{n d_{\vn}}\right)}}\Bigg[ \frac{1}{n} \log_2 \left(  f^{(n)}(\bm{\omega} \st \bm{\theta})\right) + \frac{1}{n} \log_2 \left(  a^{(n)}(\bm{\omega})\right)\Bigg]
        \end{align*}
    where for $(a)$ we used, that $\card{\mathcal{T}_{\bm{\theta}}\left( (\intmodq{q})^{n d_{\vn}}\right)}$ is polynomial in $n$. By the uniform convergence shown in Lemma \ref{lem:uniform_fn} and in \cite{hayman1956generalisation}, we can switch the limit with the supremum and the statement follows.
    The bound in \eqref{equ:lim_weightSpec} for $\mathcal{W}(\ell)$ follows in an analogous manner.
\end{proof}
Figures \ref{fig:weightspec_(3,6)_q2}, \ref{fig:weightspec_(3,6)_q3} and \ref{fig:weightspec_(3,6)_q4} show the spectral growth rate of the average weight enumerator of a random regular $(3, 6)$ \ac{LDPC} code over $\intmodq{2}$, $\intmodq{3}$ and $\intmodq{4}$, respectively, and compare it to the spectral growth rate of a random code $\code \subset (\intmodq{q})^n$ of the same rate $R$. 
Note that, for a random code $\code$, the average Lee weight enumerator $\overline{W}_{\ell}^{(n)}(\code)$, for every $\ell \in \set{0, \ldots, n\floor{q/2}}$ is computed as
\begin{align*}
    \overline{W}^{(n)}_{\ell}(\code) = \card{\code} \prob(\LW(\bfx) = \ell)
\end{align*}
Taking logarithms and limits on both sides, and defining $\delta = \ell/n$ yields
\begin{align}
    \mathcal{W}(\delta n) 
    &= \lim_{n \tendsto \infty} \frac{1}{n} \log_2 \left( \card{\code} \right) + \lim_{n \tendsto \infty} \frac{1}{n} \log_2 \left( \prob(\LW(\bfx) = \delta n) \right)\\
    &= R_2 + \lim_{n \tendsto \infty} \frac{1}{n} \log_2 \left( \prob(\bfx \in \leesphere{\delta n, q} ) \right)\\
    &= R_2 + \lim_{n \tendsto \infty} \frac{1}{n} \log_2 \left(\frac{\card{\leesphere{\delta n, q}}}{q^n}  \right)\\
    &= R_2 - \log_2(q) + H(B_{\delta}).
\end{align}

\begin{figure}
    \centering
    \begin{tikzpicture}[spy using outlines={rectangle, magnification=10, size=1.5cm, connect spies, thin},scale = 0.9, every node/.style={scale=0.85}]
\begin{axis}[
    width=0.7\columnwidth,
    height=0.5\columnwidth,
    legend pos=north east,
    grid=both,
    grid style={dotted,gray!50},
    xmin = 0,
    xmax=1,
    ymin = -0.05,
    ylabel={$\lim_{n\tendsto \infty} \frac{1}{n} \log \overline{W}_{n \delta}^{(n)}$},
    xlabel={Normalized Lee weight $\delta$},
    legend style={font=\scriptsize},
    baseline=0pt,
    ]
    \addplot[teal!90!blue,solid,mark options={solid}] table[x=ell,y=spectrum] {avrg_weightSpec_q2_6,3LDPC.txt};
    \addplot[purple,dashed,mark options={solid}] table[x=delta,y=spectrum] {Avrg_weightSpec_q2_rate_0.5_randomcode.txt};
    \draw[thin, gray] (0,0) -- (2, 0);
\end{axis}
\node[] at (0.3\columnwidth, 0.2\columnwidth) {
    \begin{tikzpicture}
    \begin{axis}[
            tiny,
            xmin=0,
            xmax=0.03,
            baseline=0pt,
        ]
        \draw[thin, gray] (0,0) --(0.1, 0);
        \addplot[teal!90!blue,smooth,mark options={solid}] table[x=ell,y=spectrum] {avrg_weightSpec_q2_6,3LDPC.txt};
    \end{axis}
    \end{tikzpicture}
    };
\draw[black, thin] (-0.1,0.4) rectangle (0.3, 0.8);
\draw[black, thin] (0.245\columnwidth, 0.162\columnwidth) -- (0.3, 0.8);
\end{tikzpicture}
    \caption{Spectral growth rate of the average weight enumerator of a regular $(3, 6)$ LDPC code ensembles over $\intmodq{2}$ (solid blue line) versus the spectral growth rate of the average weight enumerator of a random code over $\intmodq{2}$ and rate $R=1/2$ (dashed red line). The logarithm is in base $q$.}
   \label{fig:weightspec_(3,6)_q2}
\end{figure}
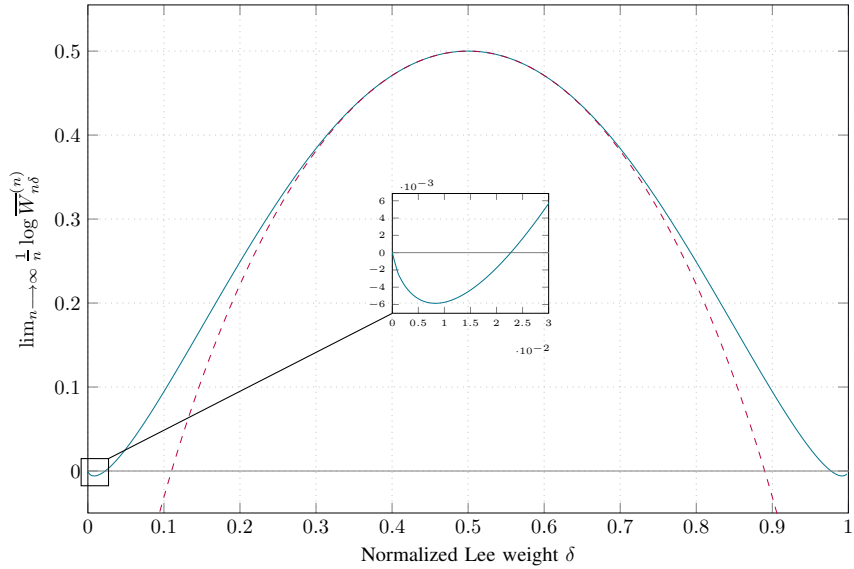
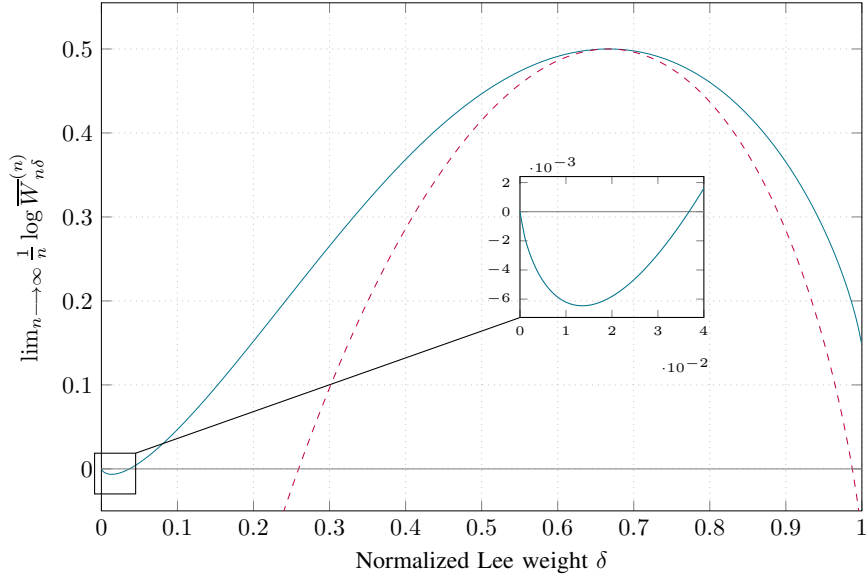
\begin{figure}
    \centering
    \begin{tikzpicture}[spy using outlines={rectangle, magnification=5, size=0.5cm, connect spies, thin},scale = 0.9, every node/.style={scale=1}]
\begin{axis}[
    width = 0.7\columnwidth,
    height = 0.5\columnwidth,
    legend pos=north east,
    grid=both,
    grid style={dotted,gray!50},
    xmin = 0,
    xmax=1,
    ymin = -0.05,
    ylabel={$\lim_{n\tendsto \infty} \frac{1}{n} \log \overline{W}_{n \delta}^{(n)}$},
    xlabel={Normalized Lee weight $\delta$},
    legend style={font=\scriptsize},
    baseline=0pt,
    ]
    \addplot[teal!90!blue,solid,mark options={solid}] table[x=ell,y=spectrum] {avrg_weightSpec_q3_6,3LDPC.txt};
    \addplot[purple,dashed,mark options={solid}] table[x=delta,y=spectrum] {Avrg_weightSpec_q3_rate_0.5_randomcode.txt};
    \draw[thin, gray] (0,0) -- (2, 0);
\end{axis}

\node[] at (0.4\columnwidth, 0.2\columnwidth) {
    \begin{tikzpicture}
    \begin{axis}[
            tiny,
            xmin=0,
            xmax=0.04,
            baseline=0pt,
        ]
        \draw[thin, gray] (0,0) --(0.1, 0);
        \addplot[teal!90!blue,smooth,mark options={solid}] table[x=ell,y=spectrum] {avrg_weightSpec_q3_6,3LDPC.txt};
    \end{axis}
    \end{tikzpicture}
    };
\draw[black, thin] (-0.1,0.25) rectangle (0.5, 0.85);
\draw[black, thin] (0.337\columnwidth, 0.157\columnwidth) -- (0.5, 0.85);
\end{tikzpicture}
    \caption{Spectral growth rate of the average weight enumerator of a regular $(3, 6)$ LDPC code ensembles over $\intmodq{3}$ (solid blue line) versus the spectral growth rate of the average weight enumerator of a random code over $\intmodq{3}$ and rate $R=1/2$ (dashed red line). The logarithm is in base $q$.}
   \label{fig:weightspec_(3,6)_q3}
\end{figure}
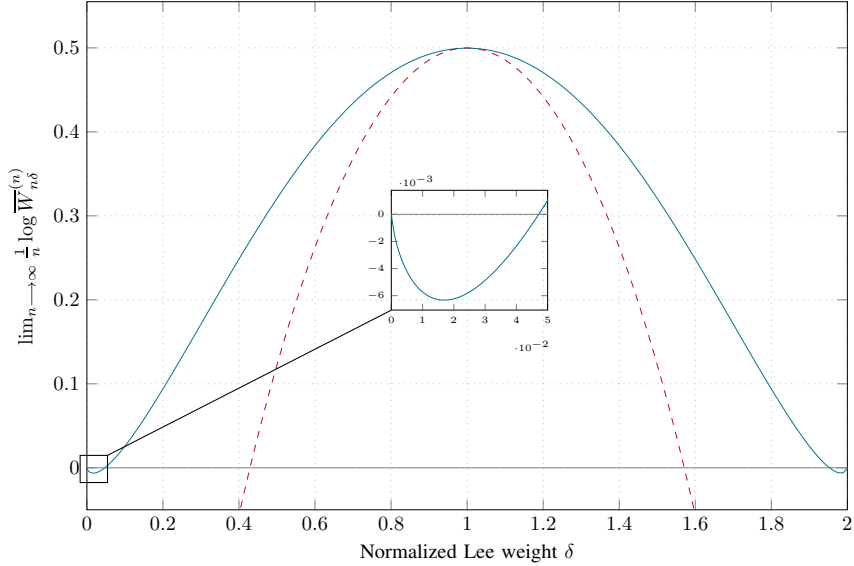
\begin{figure}
    \centering
    \begin{tikzpicture}[scale = 0.9, every node/.style={scale=0.85}]
\begin{axis}[
    width = 0.7\columnwidth,
    height =0.5\columnwidth,
    legend pos=north east,
    grid=both,
    grid style={dotted,gray!50},
    xmin = 0,
    xmax=2,
    ymin = -0.05,
    ylabel={$\lim_{n\tendsto \infty} \frac{1}{n} \log \overline{W}_{n \delta}^{(n)}$},
    xlabel={Normalized Lee weight $\delta$},
    legend style={font=\scriptsize},
    baseline=0pt,
    ]
    \addplot[teal!90!blue,solid,mark options={solid}] table[x=ell,y=spectrum] {avrg_weightSpec_q4_6,3LDPC.txt};
    \addplot[purple,dashed,mark options={solid}] table[x=delta,y=spectrum] {Avrg_weightSpec_q4_rate_0.5_randomcode.txt};
    \draw[thin, gray] (0,0) -- (2, 0);
\end{axis}

\node[] at (0.3\columnwidth, 0.2\columnwidth) {
    \begin{tikzpicture}
    \begin{axis}[
            tiny,
            xmin=0,
            xmax=0.05,
            baseline=0pt,
        ]
        \draw[thin, gray] (0,0) --(0.1, 0);
        \addplot[teal!90!blue,smooth,mark options={solid}] table[x=ell,y=spectrum] {avrg_weightSpec_q4_6,3LDPC.txt};
    \end{axis}
    \end{tikzpicture}
    };
\draw[black, thin] (-0.1,0.4) rectangle (0.3, 0.8);
\draw[black, thin] (0.245\columnwidth, 0.162\columnwidth) -- (0.3, 0.8);
\end{tikzpicture}
    \caption{Spectral growth rate of the average weight enumerator of a regular $(3, 6)$ LDPC code ensembles over $\intmodq{4}$ (solid blue line) versus the spectral growth rate of the average weight enumerator of a random code over $\intmodq{4}$ and rate $R=1/2$ (dashed red line). The logarithm is in base $q$.}
   \label{fig:weightspec_(3,6)_q4}
\end{figure}

\section{Performance Analysis of LDPC Codes over the Lee Channels}\label{sec:LDPCperformance}
In this section, we analyze the error-correction performance of regular LDPC codes over the two channel models presented in Section \ref{sec:channel}. First and foremost, we discuss an upper bound on the block error probability under \ac{ML} decoding over the {\memLC} using a union bound argument. We then focus on the performance with respect to the \ac{BP} decoder and the \ac{SMP} decoder, respectively. For both decoders we start by adapting the decoders to the Lee metric over integer residue rings discussing the main changes and assumptions needed for providing a full density evolution analysis.

\subsection{Bounds on the Block Error Probability Based on the Lee Weight Spectrum}
We are interested in the average block error probability under \ac{ML} decoding of random regular \ac{LDPC} code ensembles over $\intmodq{q}$ in the {\memLC}. As the channel is symmetric, we can assume the transmission of the zero codeword. The \ac{ML} decoder fails if and only if there is a nonzero codeword $\bfc \in \code\setminus \set{\mathbf{0}}$ satisfying
\begin{align}\label{equ:decodingFailure}
    P_{\bfY \st \bfX}(\bfy \st \bfzero) \leq P_{\bfY \st \bfX}(\bfy \st \bfc).
\end{align}
We refer to the probability of this event as the pairwise error probability and denote it by $\pep(\bfzero \rightarrow \bfc)$.
Note that in the spirit of obtaining an upper bound on the block error probability, we break ties always in favor of the erroneous codeword. Using a union bound argument, we observe that the block error probability is upper bounded by the sum of all pairwise error probabilities, i.e.,
\begin{align}\label{eq:def:decodingFailure}
    P_B(\code) \leq \sum_{\bfc \in \code\setminus \set{\bfzero}} \pep(\bfzero \rightarrow \bfc).
\end{align}
We can rewrite the pairwise error probability as 
\begin{align}\label{equ:pep_description}
    \pep(\bfzero \rightarrow \bfc) =  \prob\left( \frac{P_{\bfY \st \bfX}(\bfy \st \bfzero)}{P_{\bfY \st \bfX}(\bfy \st \bfc)} \leq 1 \right).
\end{align}
Denoting the log-likelihood ratio as
\begin{align*}
    \Lambda(y, c) := \log \left( \frac{P_{Y \st X}(y \st 0)}{P_{Y \st X}(y \st c)}\right)
\end{align*}
we have $\pep(\bfzero \rightarrow \bfc) = \prob\left( \sum_{i = 1}^n \Lambda(y_{i}, c_{i}) \leq 0 \right)$. Hence, the analysis reduces to the analysis of the distribution of the random variables $\Lambda_{\ell} := \Lambda(Y, c=\ell)$, where $Y$ is a random variable distributed as $B_\delta$. Owing to the symmetry of the Boltzmann distribution, we have that 
$$P_{Y\st X}(y \st c) = P_{Y\st X}(-y \st -c)$$
and therefore also
$$\Lambda(y, c = \ell) = \Lambda(-y, c = -\ell).$$
It follows that the distribution of $\Lambda_{\ell}$ equals the distribution of $\Lambda_{-\ell}$. Hence, the evaluation of \eqref{equ:pep_description} can be carried out by counting the number of elements in $\bfc$ possessing Lee weight $\ell$ with $\ell \in \set{0, \ldots \floor{q/2}}$. We will therefore again make use of the Lee type of a codeword (see Definition \ref{def:type_codeword}).
Thus, we can rewrite the pairwise error probability for any nonzero codeword $\bfc \in \code\setminus \set{0}$ as follows
\begin{align*}
    \pep(\bfzero \rightarrow \bfc) = \prob \left( \sum_{\ell=1}^{\floor{q/2}} \sum_{j = 1}^{n\theta_{\bfc}(\ell)} \Lambda_j \leq 0 \right).
\end{align*}
This gives us an exact value of the pairwise error probability
under \ac{ML} decoding. However, this expression requires eventually to iterate over every Lee type in the code $\code$ and is therefore inefficient for codes with large parameters. In the following we present a ``worst case'' candidate for the pairwise error probability which ultimately serves to upper bound the block error probability.

\begin{lemma}\label{lem:worstcase_PEP_smallWeight}
    Consider a nonzero codeword $\bfc \in \code$ such that $\LW(\bfc) = t $.
    Let $\bfx_{|_{t}} \in (\intmodq{q})^n$ be of Lee type $\bm{\theta}_{\bfx_{|_{t}}} = (1 - t/n, t/n, 0, \ldots , 0)$. Over a {\memLC} with $\delta\leq \delta_q$ we have
    \begin{align*}
        \pep(\bfzero \rightarrow \bfc) \leq \pep(\bfzero \rightarrow \bfx_{|_{t}}) 
    \end{align*}
    where equality holds if and only if $\bfc$ is of the same Lee type $\bm{\theta}_{\bfc} = \bm{\theta}_{\bfx_{|_{t}}}$.
\end{lemma}
Observe that the nonzero positions of $\bfx_{|_{t}}$ consist only of elements of Lee weight $1$. Therefore, it holds that $\LW(\bfx_{|_{t}}) = \HW(\bfx_{|_{t}}) \leq \LW(\bfc)$.
Figure \ref{fig:worstcase_sameLee_diffHamm} gives empirical evidence supporting the result of Lemma \ref{lem:worstcase_PEP_smallWeight}.

\begin{figure}
    \centering
    \begin{tikzpicture}[scale = 0.9, every node/.style={scale=0.85}]
    \begin{semilogyaxis}[
        width = 0.7\columnwidth,
        height = 0.5\columnwidth,
	legend pos=south east,
	grid=both,
	grid style={dotted,gray!50},
	xmax=1.2,
	ymax=5^-1,
	ylabel={Pairwise error probability},
	xlabel={Normalized Lee weight $\delta$},
	legend style={font=\normalsize},
	]
    \addplot[teal!80!blue,dotted,mark=x,mark options={solid}] table[x=Delta,y=PEP] {sameLee_diffHamm_q9_v11111111111.txt};\addlegendentry{Hamming weight $11$}
    
    \addplot[teal!60!cyan,dotted,mark=o,mark options={solid}] table[x=Delta,y=PEP] {sameLee_diffHamm_q9_v111111122.txt};\addlegendentry{Hamming weight $q9$}
    
    \addplot[green!40!teal,dotted,mark=star,mark options={solid}] table[x=Delta,y=PEP] {sameLee_diffHamm_q9_v1111124.txt};\addlegendentry{Hamming weight $7$}
    
    \addplot[purple!40!yellow,dotted,mark=triangle,mark options={solid}] table[x=Delta,y=PEP] {sameLee_diffHamm_q9_v1244.txt};\addlegendentry{Hamming weight $4$}
    
    \addplot[purple!70!orange,dotted,mark=square,mark options={solid}] table[x=Delta,y=PEP] {sameLee_diffHamm_q9_v344.txt};\addlegendentry{Hamming weight $3$}

    \end{semilogyaxis}
\end{tikzpicture}
    \caption{Comparison of the pairwise error probabilities over $\intmodq{9}$ of vectors of Lee weight $11$ and varying Hamming weight.}\label{fig:worstcase_sameLee_diffHamm}
\end{figure}
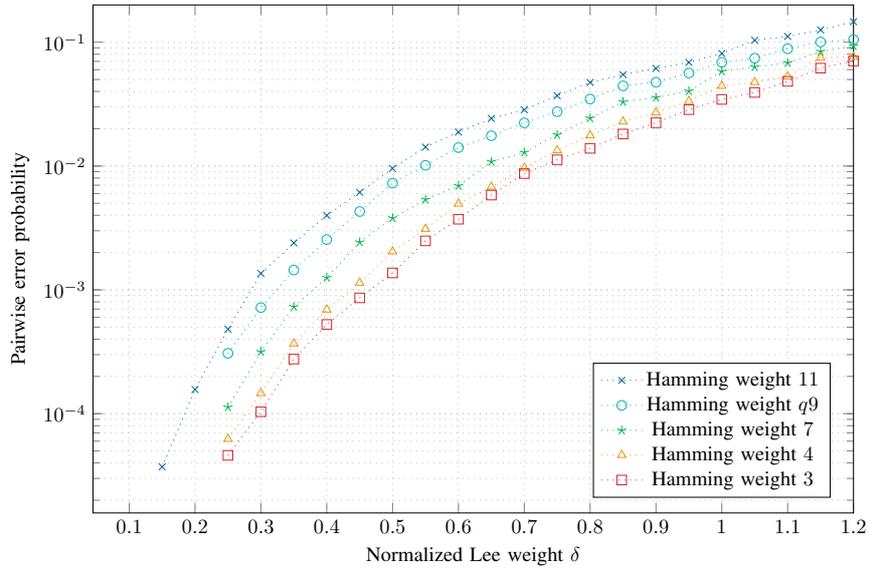

We can use these results to upper bound on the block error probability of a  linear code over the {\memLC} as a function of the codes Lee distance spectrum, for $\delta\leq \delta_q $.

\begin{corollary}\label{cor:failure_upperbound}
    Consider an $[n, k]$ linear code $\code \subset (\intmodq{q})^n$. For all $\ell \in \set{0, \ldots, n\floor{q/2}}$ let $W_{\ell}^{(n)}(\code)$ denote the Lee weight enumerator of $\code$.
    The block error probability of $\code$ under \ac{ML} decoding over the {\memLC} $\delta\leq \delta_q $ is upper bounded as
    \begin{align}
        P_B(\code)
            &\leq 
        \sum_{\ell = 1}^{n\floor{q/2}} W_{\ell}^{(n)}(\code)\prob\left( \sum_{i = 1}^{\min(\ell, n)} \Lambda_1 < 0 \right).\label{eq:UB}
    \end{align}
\end{corollary}

\begin{proof}
    Recall from \eqref{eq:def:decodingFailure} that the block error probability of $\code$ is upper bounded by the sum of all pairwise error probabilities.
    By applying Lemma \ref{lem:worstcase_PEP_smallWeight} to the $\pep$-terms, the pairwise error probability can further be upper bounded by the probability of sending the zero codeword but decoding into a word whose nonzero elements are of Lee weight $1$ only. 
\end{proof}

\begin{example}\label{ex:UB}
Figure \ref{fig:PEP_bound_plot} depicts the union bound provided in Corollary \ref{cor:failure_upperbound}, together with the block error probability estimated via Monte Carlo simulation. For the comparison we used a linear code over $\intmodq{7}$ of length $n = 6$ and dimension $k = 2$ with generator matrix
\begin{align*}
    \mathbf{G} = \begin{pmatrix} 1 & 0 & 3 & 3 & 3 & 0 \\ 0 & 1 & 0 & 4 & 3 & 3 \end{pmatrix}.
\end{align*}
As usually observed, the union bound provide accurate estimates at sufficiently low error probability.
\end{example}

\begin{figure}[]
    \centering
    \begin{tikzpicture}[scale = 0.9, every node/.style={scale=0.85}]
    \begin{semilogyaxis}[
        width=0.7\columnwidth,
        height=0.5\columnwidth,
        legend pos=south east,
        grid=both,
        grid style={dotted,gray!50},
        xmin=0.1,
        xmax=1,
        ymax=1,
        ylabel={Block error probability},
        xlabel={Normalized Lee weight $\delta$},
        legend style={font=\normalfont},
        legend cell align={left}
        ]
        
        \addplot[teal!90!blue,dotted,mark=x,mark options={solid}] table[x=Delta,y=avrgPEP] {PEP_Bound_n6_q7_k2_2.0.txt};\addlegendentry{Monte Carlo simulation}
        
        \addplot[purple!70!orange,dotted,mark=square,mark options={solid}] table[x=Delta,y=boundPEP] {PEP_Bound_n6_q7_k2_2.0.txt};\addlegendentry{Bound from Corollary \ref{cor:failure_upperbound}}

    \end{semilogyaxis}
\end{tikzpicture}
    \caption{Comparison of the union bound from Corollary \ref{cor:failure_upperbound} with respect to the performance measured via Monte Carlo simulation for the linear code over $\intmodq{7}$ of length $n = 6$ and dimension $k = 2$ from Example \ref{ex:UB}.}
    \label{fig:PEP_bound_plot}
\end{figure}
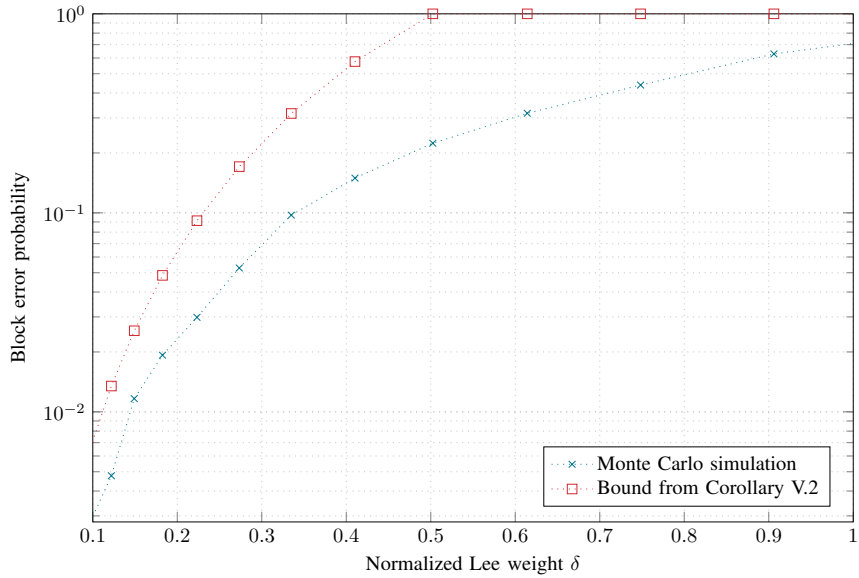

The union bound of Corollary \ref{cor:failure_upperbound} can be readily used to study the error floor performance of regular \ac{LDPC} code ensembles. To do so, it is sufficient to replace the weight enumerator $W_{\ell}^{(n)}$ in \eqref{eq:UB} with the ensemble average enumerator $\overline{W}_{\ell}^{(n)}$. An example is provided in Figure \ref{fig:LDPC_weightperformance}, where the union bound on the \ac{ML} decoding average block error probability for the $(3, 6)$-LDPC code ensemble of length $n=256$ over $\intmodq{4}$ is depicted. The result is compared with the numerical simulation for a code from the ensemble, under \ac{BP} decoding. As typical of union bounds on the block error probability, the bound is not informative above the cut-off rate of the channel. However, it provides an indication of the error probability regime at which an error floor may be observed, allowing for a quick estimation of the capability of certain code ensembles to attain given target error probabilities.

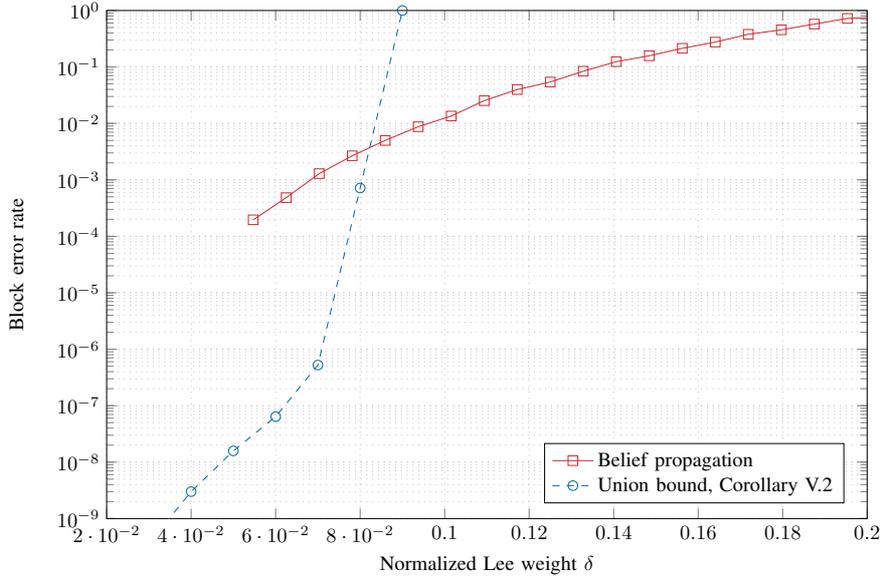
\begin{figure}
    \centering
    \begin{tikzpicture}[scale = 0.9, every node/.style={scale=0.85}]
\begin{semilogyaxis}[
    width=0.7\columnwidth,
    height=0.5\columnwidth,
    legend pos=south east,
    grid=both,
    grid style={dotted,gray!50},
    xmin=0.02,
    xmax=0.2,
    ymin=1e-9,
    ymax=1,
    xlabel={Normalized Lee weight $\delta$},
    ylabel={Block error rate},
    y label style={at={(-0.1,0.5)}},
    legend cell align=left,
    legend style={font=\normalfont},
    ]
    
    
    \addplot[purple!70!orange,solid,mark=square,mark options={solid}] table[x=NE,y=CER] {NB_BP_Z4_n256_memoryless.txt};\addlegendentry{Belief propagation}
    
    \addplot[teal!80!blue,dashed,mark=o,mark options={solid}] table[x=Delta,y=boundPEP] {2.0binomial_PEP_Bound_n256_q4.txt};\addlegendentry{Union bound, Corollary \ref{cor:failure_upperbound}}

\end{semilogyaxis}
\end{tikzpicture}
    \caption{Union bound versus belief propagation over $\intmodq{4}$ for a random $(3, 6)$ LDPC code of length $n = 256$.}
    \label{fig:LDPC_weightperformance}
\end{figure}

\subsection{Density Evolution Analysis}\label{subsec:DE}
The analysis of the Lee spectrum of LDPC code ensembles can be used, in conjunction with the union bound, to analyze the ensembles behaviour under \ac{ML} decoding at low error rates. Nevertheless, it fails to capture the block error probability behaviour in the waterfall region, under iterative decoding. We hence complement the distance spectrum analysis with a density evolution characterization of the ensemble in the limit of large block lengths.
In particular, we estimate the asymptotic iterative decoding threshold over the {\memLC} under \ac{BP} and \ac{SMP} decoding.
The iterative decoding threshold $\thr$ is defined as the largest value of the channel parameter $\delta$ for which, in the limit of large $n$ and large maximal number of iterations $\ell_\text{max}$, the symbol error probability of an \ac{LDPC} code picked randomly from a $(d_{v}, d_{c})$ code ensemble becomes vanishing small~\cite{studio3:richardson01capacity}. Owing to the complexity of tracking the evolution of the distribution of multi-dimensional messages, under \ac{BP} decoding we resort to the Monte Carlo method \cite{davey1998monte}. We denote by $\delta_{\mathsf{BP}}^{\star}$ the decoding threshold under \ac{BP} decoding.

The density evolution analysis for the \ac{SMP} decoder has been introduced in~\cite[Sec.~IV]{Lazaro19:SMP}. We will briefly sketch the idea and emphasize the respective modifications according to the new {\memLC}. For the \ac{SMP} decoder the density evolution analysis not only aims at estimating the decoding threshold $\thrSMP$ but it also provides bounds on the error probabilities $\xi$ of the extrinsic channel modelled as \ac{qSC} which are needed in the computation of the aggregated extrinsic log-likelihood vector \eqref{eq:def_E_vec}.
Since the {\memLC} is symmetric and the code is linear, we can assume that the zero codeword has been transmitted. Similar to the notation used in the description of the \ac{SMP} decoder, we let $\msg{\vn}{\cn}^{(\ell)}$ denote the message sent from variable node $\vn$ to check node $\cn$ in the $\ell$-th iteration. For every $a \in \intmodq{q}$, let us define the probability of sending a message $\msg{\vn}{\cn}^{(\ell)} = a$, knowing that originally zero has been transmitted as
\begin{align*}
    p_a^{(\ell)} := \prob \left( \msg{\vn}{\cn}^{(\ell)} = a \st X = 0 \right).
\end{align*}
Hence, recalling the {\memLC} transition probability $P_{Y \st X}(y \st x)$ from~\eqref{eq:Lee_channel}, we initialize the density evolution analysis by computing for each $a\in \intmodq{q}$ the probabilities
\begin{align*}
    p_a^{(0)} = P_{Y \st X}(a \st 0).
\end{align*}
As indicated above, except from the computation of the aggregated extrinsic likelihood vector, the remaining steps of the density evolution analysis are identical to~\cite[Sec.~IV]{Lazaro19:SMP}. In particular, we employ the \ac{qSC} approximation for the extrinsic channel.

Table~\ref{tab:dec_thresholds} records the decoding thresholds $\thrSMP$ and $\thrNBP$ for the \ac{SMP} and \ac{BP} decoder, respectively, for both $(3, 6)$ and $(4, 8)$ regular LDPC code ensembles with $q$ ranging from $5$ to $8$, as well as the Shannon limit $\thrSH$ which is given by the solution in $\delta$ of $R_2 = \log_2(q) - H(B_{\delta})$ for the rate $R_2 = 1/2$.

\renewcommand{\arraystretch}{1.2}
\begin{table}[H]
\centering
\caption{Decoding thresholds for regular \ac{LDPC} code ensembles under \ac{BP} and \ac{SMP} decoding.}
\begin{tabular}{ccccc}
 \hline\hline
 $~~~q~~~$ & $~~~(v,c)~~~$ & $~~~\thrNBP~~~$ & $~~~\thrSMP~~~$ & $~~~\thrSH~~~$
 \\ \hline\hline 
 \multirow{2}{*}{$5$} & $(3,6)$ & $0.2148$ &   $0.1039$ & \multirow{2}{*}{0.2684}
 \\
  & $(4,8)$ & $0.1802$ & $0.1200$ &
 \\\hline
 \multirow{2}{*}{$6$} & $(3,6)$ & $0.2485$ & $0.1151 $ & \multirow{2}{*}{0.3147}
 \\  
  & $(4,8)$ & $0.2217$ &   $0.1405$ &
 \\\hline
 \multirow{2}{*}{$7$} & $(3,6)$ & $0.3086$ &  $0.1261$ & \multirow{2}{*}{0.3560}
 \\
  & $(4,8)$ & $0.2686$ & $0.1539$  &
 \\\hline
 \multirow{2}{*}{$8$} & $(3,6)$ & $0.3135$ & $0.1374$ & \multirow{2}{*}{0.3950}
 \\
  & $(4,8)$ & $0.26904$ & $0.1623$ &
 \\\hline\hline
\end{tabular}
\label{tab:dec_thresholds}
\end{table}

\begin{remark}
    The choice of the \ac{DMC} used to model the extrinsic channel plays a crucial role for the \ac{SMP} algorithm, especially concerning the decoding performance. In \cite{lechner2011analysis}, for the case of \ac{BMP} decoding, it was suggested to model the \ac{VN} inbound messages as observations of a \ac{BSC}, whose transition probability was estimated by means of density evolution analysis. The approach was generalized in \cite{Lazaro19:SMP} for \ac{SMP}, where the \ac{VN} inbound messages are modelled as observations of a \ac{qSC}. In our setting we will also model the extrinsic channel as a \ac{qSC} defined in \eqref{eq:QSCapproxdec}, although in our setting the \ac{qSC} model holds only in an approximate sense.
\end{remark}

The adoption of the \ac{qSC} approximation is particularly useful from a practical viewpoint since the \ac{VN} processing in \ac{SMP} decoding becomes particularly simple if the \ac{VN}-to-\ac{CN} messages are assumed to be observations of an extrinsic \ac{qSC}. Moreover, this specific choice is motivated by the fact that, for \ac{LDPC} codes over finite fields, the extrinsic channel transition probabilities, averaged over a uniform distribution of nonzero elements in the parity-check matrix, yield (in the limit of a large block length) a \ac{qSC} \cite{Lazaro19:SMP}.
The following Lemma for $q$ prime, whose proof is trivial, supports this statement.
    \begin{lemma}\label{lem:qSC_prime}
        Consider a prime number $q$. Let $H$ be a random variable drawn uniformly at random form the multiplicative group $\units{q}$ and let $X$ be any random variable over $\intmodq{q}$. Define the random variable $V = X \cdot H$. Then $V$ follows a \ac{qSC}-like distribution given as
        \begin{align}
            \prob (V = v) 
            =
            \begin{cases}
                \prob (X = 0) & \text{if } v = 0 \\
                \frac{1}{q-1} (1 - \prob (X = 0)) & \text{else}.
            \end{cases}
        \end{align}
    \end{lemma}

Even though for $q$ is non-prime the average extrinsic channel transition probabilities can not be represented by a \ac{qSC}, we still make this assumption. 
The Empirical evidence obtained by measuring the total variation distance between the true extrinsic channel and the \ac{qSC} shows that the \ac{qSC} can still be used to accurately model the actual extrinsic channel, especially if the ring possesses relatively many unit elements. More precisely, we show numerically that the total variation distance between the two message distributions tends to zero as the number of iteration grows.
We denote by $\mathcal{U}_q$ the fraction of units in $\intmodq{q}$, i.e.,
\begin{align*}
    \mathcal{U}_q := \frac{\card{\units{q}}}{\card{\intmodq{q}}}.
\end{align*}
In order to cover different cases and support the conjecture that the \ac{qSC} assumption is especially accurate for integer rings with relatively many units, we chose three integer rings having different fractions of units. Namely, we chose $\intmodq{8}$ with $\mathcal{U}_8 = 1/2$, $\intmodq{9}$ with $\mathcal{U}_9 = 2/3$ and $\intmodq{12}$ with $\mathcal{U}_{12} = 1/3$. Figures \ref{fig:TV_(3,6)}, \ref{fig:TV_(4,8)} and \ref{fig:TV_(5,10)} show the evolution of the total variation distance with the number of iterations for different regular LDPC code ensembles, respectively. In each figure and for each integer ring, we consider three different situations: one where the relative Lee weight $\delta$ is below $\thrSMP$, one where $\delta$ is close to $\thrSMP$ and one where the relative Lee weight exceeds the threshold. The figures clearly support the conjecture on the fraction of units $\mathcal{U}_q$ as well as the choice to model the average extrinsic channel transition probabilities by a \ac{qSC}.
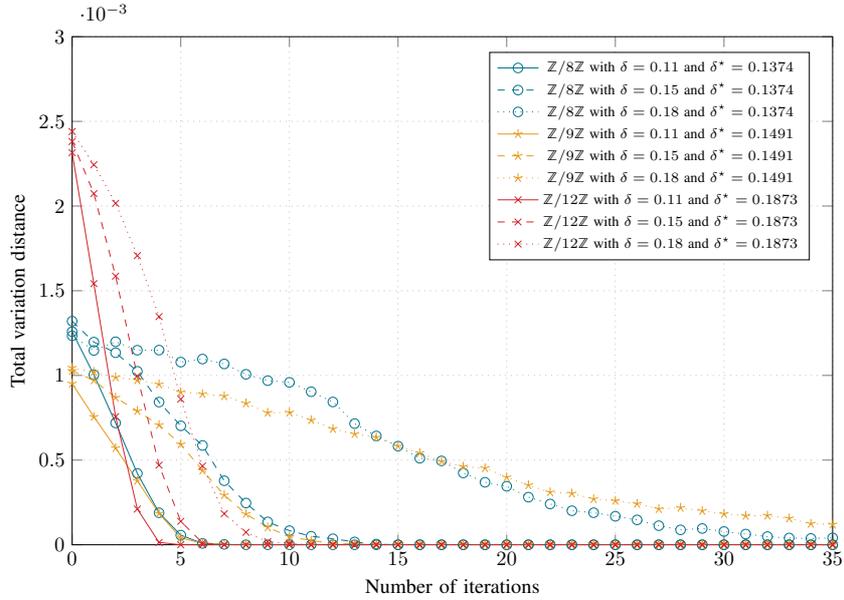
\begin{figure}[]
    \centering
    \begin{tikzpicture}[scale = 0.9, every node/.style={scale=0.85}]
    \begin{axis}[
        width=0.7\columnwidth,
	height=0.5\columnwidth,
	legend pos = north east,
	grid = both,
	grid style = {dotted,gray!50},
        xmin = 0,
	xmax = 35,
	ymin = 0,
        ymax = 3*10^-3,
	ylabel = {Total variation distance},
	xlabel = {Number of iterations},
	legend style = {font=\scriptsize},
	]
		
    \addplot[teal!90!blue,solid,mark=o,mark options={solid}] table[x=Iter,y=TV] {TV_3,6_delta0.11_q8.txt};\addlegendentry{$\intmodq{8}$ with $\delta=0.11$ and $\delta^{\star} = 0.1374$}
    \addplot[teal!90!blue,dashed,mark=o,mark options={solid}] table[x=Iter,y=TV] {TV_3,6_delta0.15_q8.txt};\addlegendentry{$\intmodq{8}$ with $\delta=0.15$ and $\delta^{\star} = 0.1374$}
    \addplot[teal!90!blue,dotted,mark=o,mark options={solid}] table[x=Iter,y=TV] {TV_3,6_delta0.18_q8.txt};\addlegendentry{$\intmodq{8}$ with $\delta=0.18$ and $\delta^{\star} = 0.1374$}

    \addplot[purple!40!yellow,solid,mark=star,mark options={solid}] table[x=Iter,y=TV] {TV_3,6_delta0.11_q9.txt};\addlegendentry{$\intmodq{9}$ with $\delta=0.11$ and $\delta^{\star} = 0.1491$}
    \addplot[purple!40!yellow,dashed,mark=star,mark options={solid}] table[x=Iter,y=TV] {TV_3,6_delta0.15_q9.txt};\addlegendentry{$\intmodq{9}$ with $\delta=0.15$ and $\delta^{\star} = 0.1491$}
    \addplot[purple!40!yellow,dotted,mark=star,mark options={solid}] table[x=Iter,y=TV] {TV_3,6_delta0.18_q9.txt};\addlegendentry{$\intmodq{9}$ with $\delta=0.18$ and $\delta^{\star} = 0.1491$}

    \addplot[purple!70!orange,solid,mark=x,mark options={solid}] table[x=Iter,y=TV] {TV_3,6_delta0.11_q12.txt};\addlegendentry{$\intmodq{12}$ with $\delta=0.11$ and $\delta^{\star} = 0.1873$}
    \addplot[purple!70!orange,dashed,mark=x,mark options={solid}] table[x=Iter,y=TV] {TV_3,6_delta0.15_q12.txt};\addlegendentry{$\intmodq{12}$ with $\delta=0.15$ and $\delta^{\star} = 0.1873$}
    \addplot[purple!70!orange,dotted,mark=x,mark options={solid}] table[x=Iter,y=TV] {TV_3,6_delta0.18_q12.txt};\addlegendentry{$\intmodq{12}$ with $\delta=0.18$ and $\delta^{\star} = 0.1873$}

    \end{axis}
\end{tikzpicture}
    \caption{Evolution of the \ac{TV} distance between the extrinsic channel distribution and the \ac{qSC} for regular $(3, 6)$ LDPC code ensembles in the \ac{SMP} decoder.}
    \label{fig:TV_(3,6)}
\end{figure}
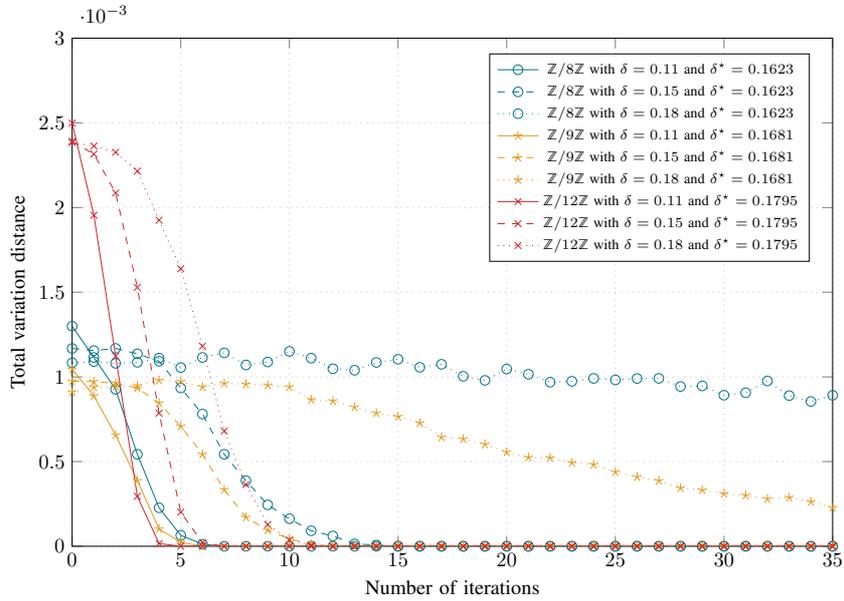
\begin{figure}[]    
    \centering
    \begin{tikzpicture}[scale = 0.9, every node/.style={scale=0.85}]
    \begin{axis}[
        width=0.7\columnwidth,
	height=0.5\columnwidth,
	legend pos=north east,
	grid=both,
	grid style={dotted,gray!50},
        xmin = 0,
	xmax=35,
	ymin=0,
        ymax = 3*10^-3,
	ylabel={Total variation distance},
	xlabel={Number of iterations},
	legend style={font=\scriptsize},
	]
		
    \addplot[teal!90!blue,solid,mark=o,mark options={solid}] table[x=Iter,y=TV] {TV_4,8_delta0.11_q8.txt};\addlegendentry{$\intmodq{8}$ with $\delta=0.11$ and $\delta^{\star} = 0.1623$}
    \addplot[teal!90!blue,dashed,mark=o,mark options={solid}] table[x=Iter,y=TV] {TV_4,8_delta0.15_q8.txt};\addlegendentry{$\intmodq{8}$ with $\delta=0.15$ and $\delta^{\star} = 0.1623$}
    \addplot[teal!90!blue,dotted,mark=o,mark options={solid}] table[x=Iter,y=TV] {TV_4,8_delta0.18_q8.txt};\addlegendentry{$\intmodq{8}$ with $\delta=0.18$ and $\delta^{\star} = 0.1623$}

    \addplot[purple!40!yellow,solid,mark=star,mark options={solid}] table[x=Iter,y=TV] {TV_4,8_delta0.11_q9.txt};\addlegendentry{$\intmodq{9}$ with $\delta=0.11$ and $\delta^{\star} = 0.1681$}
    \addplot[purple!40!yellow,dashed,mark=star,mark options={solid}] table[x=Iter,y=TV] {TV_4,8_delta0.15_q9.txt};\addlegendentry{$\intmodq{9}$ with $\delta=0.15$ and $\delta^{\star} = 0.1681$}
    \addplot[purple!40!yellow,dotted,mark=star,mark options={solid}] table[x=Iter,y=TV] {TV_4,8_delta0.18_q9.txt};\addlegendentry{$\intmodq{9}$ with $\delta=0.18$ and $\delta^{\star} = 0.1681$}

    \addplot[purple!70!orange,solid,mark=x,mark options={solid}] table[x=Iter,y=TV] {TV_4,8_delta0.11_q12.txt};\addlegendentry{$\intmodq{12}$ with $\delta=0.11$ and $\delta^{\star} = 0.1795$}
    \addplot[purple!70!orange,dashed,mark=x,mark options={solid}] table[x=Iter,y=TV] {TV_4,8_delta0.15_q12.txt};\addlegendentry{$\intmodq{12}$ with $\delta=0.15$ and $\delta^{\star} = 0.1795$}
    \addplot[purple!70!orange,dotted,mark=x,mark options={solid}] table[x=Iter,y=TV] {TV_4,8_delta0.18_q12.txt};\addlegendentry{$\intmodq{12}$ with $\delta=0.18$ and $\delta^{\star} = 0.1795$}

    \end{axis}
\end{tikzpicture}
    \caption{Evolution of the \ac{TV} distance between the extrinsic channel distribution and the \ac{qSC} for regular $(4, 8)$ LDPC code ensembles in the \ac{SMP} decoder.}
    \label{fig:TV_(4,8)}
\end{figure}
\begin{figure}[]   
    \centering
    \begin{tikzpicture}[scale = 0.9, every node/.style={scale=0.85}]
    \begin{axis}[
        width=0.7\columnwidth,
	height=0.5\columnwidth,
	legend pos=north east,
	grid=both,
	grid style={dotted,gray!50},
        xmin = 0,
	xmax=35,
	ymin=0,
        ymax = 3*10^-3,
	ylabel={Total variation distance},
	xlabel={Number of iterations},
	legend style={font=\scriptsize},
	]
		
    \addplot[teal!90!blue,solid,mark=o,mark options={solid}] table[x=Iter,y=TV] {TV_5,10_delta0.11_q8.txt};\addlegendentry{$\intmodq{8}$ with $\delta=0.11$ and $\delta^{\star} = 0.1437$}
    \addplot[teal!90!blue,dashed,mark=o,mark options={solid}] table[x=Iter,y=TV] {TV_5,10_delta0.15_q8.txt};\addlegendentry{$\intmodq{8}$ with $\delta=0.15$ and $\delta^{\star} = 0.1437$}
    \addplot[teal!90!blue,dotted,mark=o,mark options={solid}] table[x=Iter,y=TV] {TV_5,10_delta0.18_q8.txt};\addlegendentry{$\intmodq{8}$ with $\delta=0.18$ and $\delta^{\star} = 0.1437$}

    \addplot[purple!40!yellow,solid,mark=star,mark options={solid}] table[x=Iter,y=TV] {TV_5,10_delta0.11_q9.txt};\addlegendentry{$\intmodq{9}$ with $\delta=0.11$ and $\delta^{\star} = 0.1553$}
    \addplot[purple!40!yellow,dashed,mark=star,mark options={solid}] table[x=Iter,y=TV] {TV_5,10_delta0.15_q9.txt};\addlegendentry{$\intmodq{9}$ with $\delta=0.15$ and $\delta^{\star} = 0.1553$}
    \addplot[purple!40!yellow,dotted,mark=star,mark options={solid}] table[x=Iter,y=TV] {TV_5,10_delta0.18_q9.txt};\addlegendentry{$\intmodq{9}$ with $\delta=0.18$ and $\delta^{\star} = 0.1553$}

    \addplot[purple!70!orange,solid,mark=x,mark options={solid}] table[x=Iter,y=TV] {TV_5,10_delta0.11_q12.txt};\addlegendentry{$\intmodq{12}$ with $\delta=0.11$ and $\delta^{\star} = 0.1760$}
    \addplot[purple!70!orange,dashed,mark=x,mark options={solid}] table[x=Iter,y=TV] {TV_5,10_delta0.15_q12.txt};\addlegendentry{$\intmodq{12}$ with $\delta=0.15$ and $\delta^{\star} = 0.1760$}

    \end{axis}
\end{tikzpicture}
    \caption{Evolution of the \ac{TV} distance between the extrinsic channel distribution and the \ac{qSC} for regular $(5, 10)$ LDPC code ensembles in the \ac{SMP} decoder.}
    \label{fig:TV_(5,10)}
\end{figure}
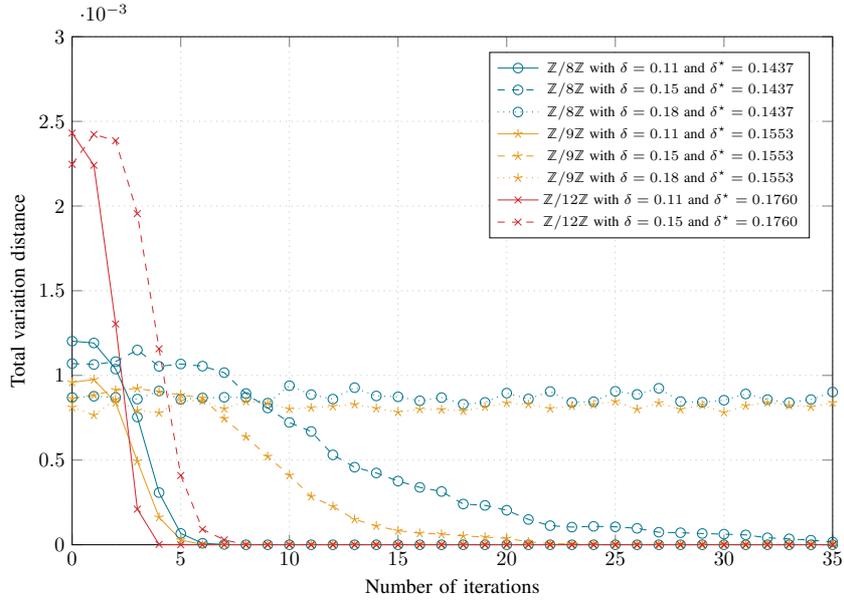

\subsection{Numerical Results}\label{subsec:results}
We finally present numerical results showing the decoding performance (in terms of block error rates) of $(3, 6)$ regular LDPC codes of length $n = 256$ under both \ac{BP} and \ac{SMP} decoding. We chose to analyze the performances over three different integer rings, namely $\intmodq{5}$, $\intmodq{7}$ and $\intmodq{8}$. The performances will additionally be compared to the \ac{LSF} decoder presented in \cite[Algorithm 2]{santini2020low}. Following the suggestions of \cite{santini2020low}, we assumed a decoding threshold $\tau = \frac{d_{\vn}}{2}$ for the \ac{LSF} decoder. All the results were obtained using Monte Carlo simulations. The codes used in the simulations have been obtained by first generating their bipartite graphs via the \ac{PEG} algorithm \cite{HEA05}, and then by assigning to the nonzero entries in the code parity-check matrices elements sampled uniformly at random and independently from $\units{q}$. The error vectors in the {\constLC} are drawn uniformly at random from the Lee sphere of a given radius representing the desired weight according to \cite[Algorithms 1 and 2]{bariffi2022properties}, whereas in the {\memLC} the entries of the error vector are drawn according to the distribution defined in \eqref{eq:Lee_channel}. In  both cases, the performance is compared to the \ac{RCU} bounds established in Corollary \ref{cor:RCU_const_ML} and Theorem \ref{thm:RCU_mmless_ML}, respectively.\\

The block error probability evaluated over the memoryless channel is shown in Figure \ref{fig:Lee256}. The RCU bounds (dotted in the graph) show clearly the impact of the size $q$ of the finite integer ring, i.e., larger $q$ admit a larger relative Lee weight $\delta$. This is also observed in the performance under both \ac{BP} and \ac{SMP} decoding as well as in the \ac{LSF} decoder. The impact of $q$ in the \ac{SMP} is not only important for the admissible choices of $\delta$. Moreover it shows clearly the difference between $q$ prime and not. While a small gain is achieved when considering $\intmodq{8}$ instead of $\intmodq{7}$ under \ac{BP} decoding, the performance slightly suffers under \ac{SMP} decoding meaning there is almost no gain. This might be due to the \ac{qSC} assumption which holds only in an asymptotic sense for the non-field case, as discussed in Section \ref{subsec:DE}.

We observe the same effect in the performance over the {\constLC} in Figure \ref{fig:CWLee}, i.e. there is almost no gain visible when moving from $q = 7$ to $q=8$ under the \ac{SMP} decoder. Analogous to the memoryless case, we observe the same impact of the size of $\intmodq{q}$ on the possible choices of $\delta$ which is captured by the RCU bound for the \constLC. 
In both channel models we observe that the \ac{SMP} decoder outperforms the \ac{LSF} decoder despite the \ac{qSC} assumption in the extrinsic channel of the \ac{SMP}. We want to emphasize and acknowledge here that the \ac{LSF} was originally designed for low-Lee-density parity-check codes which form a special class of LDPC codes. Hence, when comparing the performances over the two decoders the difference of the code classes might be taken in consideration. Nevertheless, we will not focus deeper on this argument and leave this subject to future investigations. We believe that the additional knowledge about the marginal distribution plays a crucial part in the performance gain under \ac{SMP} decoding. Observe that the estimated threshold values obtained via density evolution analysis and stored in Table \ref{tab:dec_thresholds} match well to the actual block error rates achieved by both \ac{BP} and \ac{SMP} decoding. As expected from the predictions in Table \ref{tab:dec_thresholds}, \ac{BP} clearly outperforms \ac{SMP} decoding. However, the \ac{SMP} algorithm shows a performance that is appealing for applications demanding low-complexity decoding \cite{santini2020low}.

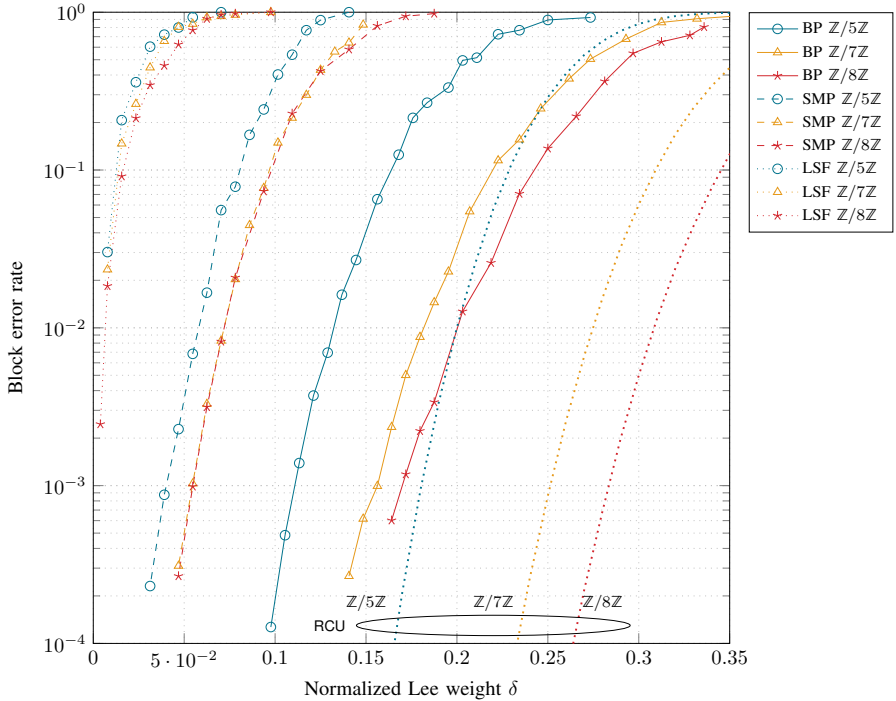
\begin{figure}
    \centering
    \begin{tikzpicture}[scale = 0.9, every node/.style={scale=0.85}]
    \begin{semilogyaxis}[
    width=0.6\columnwidth,
    height=0.6\columnwidth,
    legend pos=south west,
    grid=both,
    grid style={dotted,gray!50},
    xmin=0,
    xmax=0.35,
    ymin=1e-4,
    ymax=1,
    xlabel={Normalized Lee weight $\delta$},
    ylabel={Block error rate},
    y label style={at={(-0.1,0.5)}},
    legend pos = outer north east,
    legend cell align=left,
    legend style={font=\footnotesize},
    ] 
    
    \addplot[teal!90!blue,solid,mark=o,mark options={solid}] table[x=NE,y=CER] {NB_BP_GF5_n256_lee.txt};\addlegendentry{BP $\mathbb{Z}/5\mathbb{Z}$}
    \addplot[purple!40!yellow,solid,mark=triangle,mark options={solid}] table[x=NE,y=CER] {NB_BP_GF7_n256_lee.txt};\addlegendentry{BP $\mathbb{Z}/7\mathbb{Z}$}
    \addplot[purple!70!orange,solid,mark=star,mark=star,mark options={solid}] table[x=NE,y=CER] {NB_BP_Z8_n256_lee.txt};\addlegendentry{BP $\mathbb{Z}/8\mathbb{Z}$}
    
    \addplot[teal!90!blue,dashed,mark=o,mark options={solid}] table[x=NE,y=CER] {SMP_GF5_n256_lee.txt};\addlegendentry{SMP $\mathbb{Z}/5\mathbb{Z}$}
    \addplot[purple!40!yellow,dashed,mark=triangle,mark options={solid}] table[x=NE,y=CER] {SMP_GF7_n256_lee.txt};\addlegendentry{SMP $\mathbb{Z}/7\mathbb{Z}$}
    \addplot[purple!70!orange,dashed,mark=star,mark options={solid}] table[x=NE,y=CER] {SMP_Z8_n256_lee.txt};\addlegendentry{SMP $\mathbb{Z}/8\mathbb{Z}$}

    \addplot[teal!90!blue,dotted,mark=o,mark options={solid}] table[x=NE,y=CER] {LSF_Z5_n256_lee.txt};\addlegendentry{LSF $\mathbb{Z}/5\mathbb{Z}$}
    \addplot[purple!40!yellow,dotted,mark=triangle,mark options={solid}] table[x=NE,y=CER] {LSF_Z7_n256_lee.txt};\addlegendentry{LSF $\mathbb{Z}/7\mathbb{Z}$}
    \addplot[purple!70!orange,dotted,mark=star,mark options={solid}] table[x=NE,y=CER] {LSF_Z8_n256_lee.txt};\addlegendentry{LSF $\mathbb{Z}/8\mathbb{Z}$}

    \addplot[teal!90!blue,dotted,thick] table[x=delta,y=RCU] {RCUSPB_n256_q5.txt};
    \addplot[purple!40!yellow,dotted,thick] table[x=delta,y=RCU] {RCUSPB_n256_q7.txt};
    \addplot[purple!70!orange,dotted,thick] table[x=delta,y=RCU] {RCUSPB_n256_q8.txt};
    
    \draw[very thin] (axis cs: 0.22, 1.3e-4) ellipse (2cm and 0.15 cm);
    \node[align=center,scale=0.9] at (axis cs: 0.13, 1.3e-4) {\footnotesize \textsf{RCU}};
    \node at (axis cs: 0.15, 1.8e-4) {\footnotesize $\mathbb{Z}/5\mathbb{Z}$};
    \node at (axis cs: 0.22, 1.8e-4) {\footnotesize $\mathbb{Z}/7\mathbb{Z}$};
    \node at (axis cs: 0.28, 1.8e-4) {\footnotesize $\mathbb{Z}/8\mathbb{Z}$};	
\end{semilogyaxis}
\end{tikzpicture}

    \caption{Block error rate vs. $\delta$ for regular $(3, 6)$ nonbinary \ac{LDPC} codes of length $n = 256$, {\memLC}. \ac{LSF} compared to the \ac{RCU} bound from Theorem~\ref{thm:RCU_mmless_ML},
    \ac{SMP} and \ac{BP} decoding.}\label{fig:Lee256}
\end{figure}
\begin{figure}
    \centering
    \begin{tikzpicture}[scale = 0.9, every node/.style={scale=0.85}]
    \begin{semilogyaxis}[
        width=0.6\columnwidth,
	height=0.6\columnwidth,
        grid=both,
        grid style={dotted,gray!50},
        xmin=0,
        xmax=0.41,
        ymin=1e-4,
        ymax=1,
        xlabel={Normalized Lee weight $\delta$},
        ylabel={Block error rate},
        y label style={at={(-0.1,0.5)}},
        legend cell align=left,
        legend pos = outer north east,
        legend style={font=\footnotesize},
    ]
    
    \addplot[teal!90!blue,solid,mark=o,mark options={solid}] table[x=NE,y=CER] {NB_BP_GF5_n256_constant.txt};\addlegendentry{BP $\mathbb{Z}/5\mathbb{Z}$}
    \addplot[purple!40!yellow,solid,mark=triangle,mark options={solid}] table[x=NE,y=CER] {NB_BP_GF7_n256_constant.txt};\addlegendentry{BP $\mathbb{Z}/7\mathbb{Z}$}
    \addplot[purple!70!orange,solid,mark=star,mark=star,mark options={solid}] table[x=NE,y=CER] {NB_BP_Z8_n256_constant.txt};\addlegendentry{BP $\mathbb{Z}/8\mathbb{Z}$}
   
    \addplot[teal!90!blue,dashed,mark=o,mark options={solid}] table[x=NE,y=CER] {SMP_GF5_n256_constant.txt};\addlegendentry{SMP $\mathbb{Z}/5\mathbb{Z}$}
    \addplot[purple!40!yellow,dashed,mark=triangle,mark options={solid}] table[x=NE,y=CER] {SMP_GF7_n256_constant.txt};\addlegendentry{SMP $\mathbb{Z}/7\mathbb{Z}$}
    \addplot[purple!70!orange,dashed,mark=star,mark options={solid}] table[x=NE,y=CER] {SMP_Z8_n256_constant.txt};\addlegendentry{SMP $\mathbb{Z}/8\mathbb{Z}$}

    \addplot[teal!90!blue,dotted,mark=o,mark options={solid}] table[x=NE,y=CER] {LSF_Z5_n256_constant.txt};\addlegendentry{LSF $\mathbb{Z}/5\mathbb{Z}$}
    \addplot[purple!40!yellow,dotted,mark=triangle,mark options={solid}] table[x=NE,y=CER] {LSF_Z7_n256_constant.txt};\addlegendentry{LSF $\mathbb{Z}/7\mathbb{Z}$}
    \addplot[purple!70!orange,dotted,mark=star,mark=star,mark options={solid}] table[x=NE,y=CER] {LSF_Z8_n256_constant.txt};\addlegendentry{LSF $\mathbb{Z}/8\mathbb{Z}$}

    \addplot[teal!90!blue,dotted,thick] table[x=delta,y=RCU] {RCUConst256q5.txt};
    \addplot[purple!40!yellow,dotted,thick] table[x=delta,y=RCU] {RCUConst256q7.txt};
    \addplot[purple!70!orange,dotted,thick] table[x=delta,y=RCU] {RCUConst256q8.txt};
    
    \draw[very thin] (axis cs: 0.324, 1.3e-4) ellipse (2cm and 0.15 cm);
    \node[align=center,scale=0.9] at (axis cs: 0.22, 1.3e-4) {\footnotesize \textsf{RCU}};
    \node at (axis cs: 0.24, 1.85e-4) {\footnotesize $\mathbb{Z}/5\mathbb{Z}$};
    \node at (axis cs: 0.33, 1.85e-4) {\footnotesize $\mathbb{Z}/7\mathbb{Z}$};
    \node at (axis cs: 0.37, 1.85e-4) {\footnotesize $\mathbb{Z}/8\mathbb{Z}$};	
\end{semilogyaxis}
\end{tikzpicture}
    \caption{Block error rate vs. $\delta$ for regular $(3, 6)$ nonbinary \ac{LDPC} code ensembles of length $n = 256$ and rate $R = 1/2$ on the {\constLC} under \ac{LSF}, \ac{SMP} and \ac{BP} decoding compared to the \ac{RCU} bound from Theorem \ref{thm:RCU_const_ML}.}\label{fig:CWLee}
\end{figure}
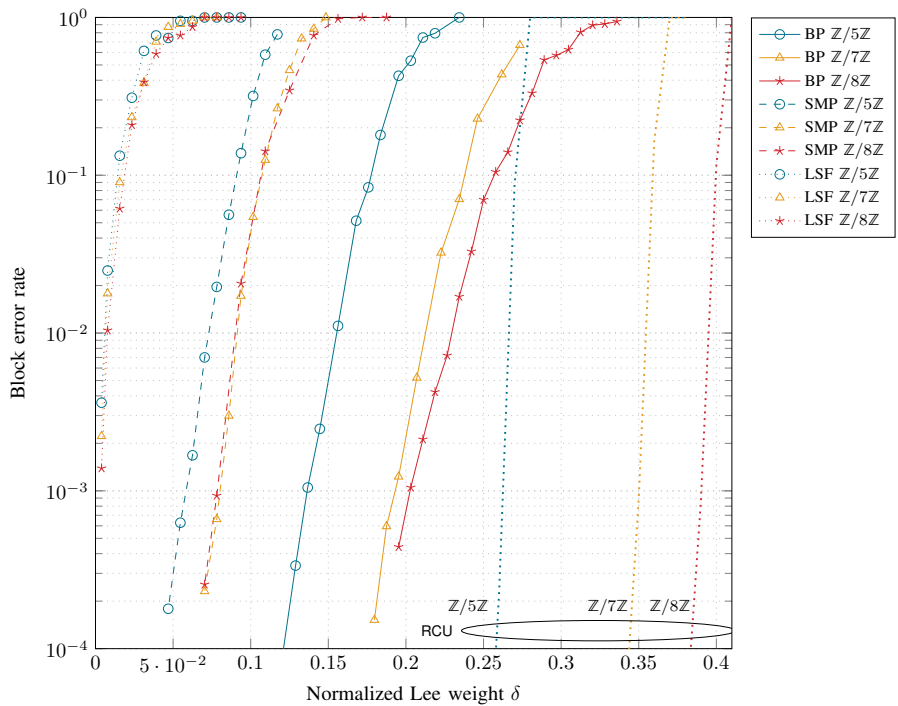

\section{Conclusions}\label{sec:conc}
In this paper we studied the decoding performance of random regular \acf{LDPC} codes over finite integer rings considering two channel models in the Lee metric, a memoryless channel model and a channel introducing an error of given Lee weight. 
We established the growth rate spectra of the Lee sphere and Lee volume, respectively. 
These results were used to derive random coding union bounds for the block error probability under maximum likelihood and minimum distance decoding for both channel models. In the case of the  {\memLC} we also derived a lower bound (in terms of a sphere packing bound) on the error probability. 
An upper bound on the \ac{ML} block error probability of linear codes based on the Lee weight enumerator of the code was introduced. The bound has been used to study the average block error probability of regular LDPC code ensembles, thanks to a derivation of the average Lee weight spectrum of the ensembles. The bound provides relevant information on the code performance in the low error probability regime (i.e., in the error floor region). The study has been complemented with a density evolution analysis. Two decoders have been considered: one based on the (non-binary) belief propagation algorithm, and a low-complexity message-passing algorithm where exchanged messages are hard symbols (i.e., ring elements).
The simulation results confirmed the outcomes of the density evolution analysis, that is belief propagation decoding outperforms symbol message passing decoding. Nevertheless, the performance under symbol message passing decoding seems a promising option for applications asking for low complexity (such as code-based cryptosystems involving the Lee metric). Furthermore, the performance analysis of regular \ac{LDPC} codes over integer residue rings that we developed might be useful to design such \ac{LDPC} codes for applications to cryptography.

In this work, we restricted ourselves to regular LDPC codes over integer residue rings. Future work includes the performance study of other families of LDPC codes over finite integer rings such as protograph-based and irregular LDPC codes.

\bibliographystyle{IEEEtran}
\bibliography{IEEEabrv, library}

\vspace{-5mm}
\begin{IEEEbiographynophoto}{Jessica Bariffi}
    (IEEE Student Member) was born in Zurich, Switzerland in 1995. She received her M.Sc. and Ph.D. degrees in mathematics from the University of Zurich, Switzerland, in 2020 and 2024, respectively. In her dissertation (supervised by Prof. Joachim Rosenthal) she studied the algebraic structure of Lee-metric codes, the decoding performance of LDPC codes over suitably designed channels in the Lee metric, as well as Information Set Decoding in the Lee metric. She pursued her Ph.D. together with the German Aerospace Center (DLR) in Munich, Germany, where she worked in the Quantum-Resistant Cryptography group under Dr. Hannes Bartz. In July 2024 she started a postdoc position at the Technical University of Munich, Germany, at the Institute for Communications Engineering under the supervision of Prof. Antonia Wachter-Zeh.
\end{IEEEbiographynophoto}
\vspace{-10mm}
\begin{IEEEbiographynophoto}{Hannes Bartz}
    (S'14-M'16) was born in Trostberg, Germany, in 1985. He received his Dipl.-Ing. and Dr.-Ing. degree from the Technical University of Munich, Germany, in 2010 and 2017, respectively. In his dissertation (supervised by Prof. Gerhard Kramer) he developed efficient algebraic decoding schemes for error-correcting codes in subspace and rank metric. In July 2017 he joined the Information Transmission Group within the Institute of Communications and Navigation at the German Aerospace Center (DLR). Since April 2021 he is leading the Quantum-Resistant Cryptography (QRC) group within the Satellite Networks department. His main research interests are code-based post-quantum cryptography and algebraic coding theory. In 2018 he has been appointed as a Lecturer at the Institute for Communications Engineering (LNT), Technical University of Munich, Germany. He received the Prof. Dr. Ralf Kötter memorial award in 2012.
\end{IEEEbiographynophoto}
\vspace{-10mm}
\begin{IEEEbiographynophoto}{Gianluigi Liva}
    (M’08–SM’14) was born in Spilimbergo, Italy, in 1977. He received the M.S. and Ph.D. degrees in electrical engineering from the University of Bologna, Italy, in 2002 and 2006, respectively. Since 2003 he has been investigating channel codes for high-data rate Consultative Committee for Space Data Systems (CCSDS) missions. From 2004 to 2005, he was involved in research at the University of Arizona, Tucson. Since 2006, he has been with the Institute of Communications and Navigation, German Aerospace Center (DLR), where he currently leads the Information Transmission Group. In 2010, he has been appointed as a Lecturer of channel coding with the Institute for Communications Engineering (LNT), Technical University of Munich (TUM). From 2012 to 2013, he was a Lecturer of channel coding with the Nanjing University of Science and Technology, China. Since 2014, he has been a Lecturer of channel codes with iterative decoding with LNT, TUM. He received the Italian National Scientific Habilitation (ASN) as Full Professor in Telecommunication Engineering in July 2017. His main research interests include satellite communications, random access techniques, and error control coding. He is/has been active in the DVB-SH, DVB-RCS, and DVB-S2 standardization groups, and in the standardization of error correcting codes for deep-space communications within the CCSDS. He was the co-chair of the 2018 IEEE European School on Information Theory, the sponsor co-chair of the IEEE Information Theory Workshop 2020 in Riva del Garda, and the TPC co-chair of the 2023 International Symposium on Topics in Coding. Since 2020, he serves as Associate Editor in Coding and Information Theory for the IEEE Transactions on Communications.
\end{IEEEbiographynophoto}
\vspace{-10mm}

\begin{IEEEbiographynophoto}{Joachim Rosenthal}
    (Fellow, IEEE) received the Diploma degree in mathematics from the University of Basel in 1986 and the Ph.D. degree in mathematics from Arizona State University in 1990. From 1990 to 2006, he was with the University of Notre Dame, USA, where he was the holder of an Endowed Chair of applied mathematics and he was also a Concurrent Professor of electrical engineering. Since 2004, he has been a Professor of applied mathematics with the University of Zurich, where was the past Chair of the Institute of Mathematics and the Vice Dean of the College of Science. He will be (2024–2025) the President of the Swiss Mathematical Society. His current research interests include coding theory and cryptography. He has served as an organizer or the program chair for numerous international conferences, e.g., he is one of the Technical Program Chair of ISIT 2024, Athens, Greece. In addition, he has served on numerous editorial boards.
\end{IEEEbiographynophoto}

\end{document}